\documentclass[9pt,sigconf,table]{acmart}

\AtBeginDocument{%
  \providecommand\BibTeX{{%
    \normalfont B\kern-0.5em{\scshape i\kern-0.25em b}\kern-0.8em\TeX}}}

\usepackage[font=footnotesize]{subcaption}
\usepackage{caption}
\usepackage{arydshln}
\usepackage{array,multirow,graphicx}
\usepackage{paralist}
\usepackage{dot2texi}
\usepackage{tikz}
\usepackage{makecell}
\usepackage{tikz-qtree,tikz-qtree-compat}
\usetikzlibrary{shapes,arrows,patterns}
\usepackage{amsmath}
\usepackage{mathtools}
\usepackage[linesnumbered,algoruled,commentsnumbered,boxruled]{algorithm2e}
\usepackage{hhline}
\usepackage{todonotes}
\usepackage{url}

\newtheorem{openpb}{Open problem}

\newcommand{\pqe}{\mathrm{PQE}}

\newcommand\sat{\mathsf{SAT}}
\newcommand\ssat{\#\mathsf{SAT}}
\newcommand\shap{\mathsf{Shapley}}

\newcommand\vars{\mathsf{Vars}}
\newcommand\ar{\mathsf{ar}}
\newcommand\const{\mathsf{Const}}

\newcommand\out{\mathsf{output}}

\newcommand\R{\mathbb{R}}
\newcommand\cprox{\textsf{CNF Proxy}}
\newcommand\mc{\textsf{Monte Carlo}}
\newcommand\kshap{\textsf{Kernel SHAP}}

\newcommand\lin{\mathsf{Lin}}
\newcommand\elin{\mathsf{ELin}}
\renewcommand\phi\varphi

\newcommand{\hrulealg}[0]{\vspace{1mm} \hrule \vspace{1mm}}
\newcommand\defeq{\stackrel{\mathrm{def}}{=}}

\newcommand{\dn}{D_\mathrm{n}}
\newcommand{\dx}{D_\mathrm{x}}

\newcommand{\tr}{\leq_{\mathrm{T}}^{\mathrm{p}}}

\newcommand{\shslices}{\text{$\#$}\mathrm{Slices}}

\newcommand*\conj[1]{\overline{#1}}
\newcommand{\paratitle}[1]{\noindent{\bf #1}}

\newcommand\fptoshp{\mathrm{FP}^{\text{$\#$}\text{\rm P}}}

\newcommand{\revision}[1]{{#1}}

\begin{document}

\title{Computing the Shapley Value of Facts in Query Answering}

\pagenumbering{gobble}
\pagenumbering{arabic}

\settopmatter{authorsperrow=4}
\settopmatter{printacmref=false}
\renewcommand\footnotetextcopyrightpermission[1]{}
\author{Daniel Deutch}
\affiliation{%
  \institution{Tel Aviv University\\
  Blavatnik School of Computer Science,
  Tel Aviv, Israel}
}
\email{danielde@post.tau.ac.il}

\author{Nave Frost}
\affiliation{%
  \institution{Tel Aviv University\\
  Blavatnik School of Computer Science,
  Tel Aviv, Israel}
}
\email{navefrost@mail.tau.ac.il}

\author{Benny Kimelfeld}
\affiliation{%
  \institution{Technion - Israel Institute of Technology \\ 
  Faculty of Computer Science, Haifa, Israel}
}
\email{bennyk@cs.technion.ac.il}

\author{Mika\"el Monet}
\affiliation{%
  \institution{Univ. Lille, Inria, CNRS, Centrale Lille, UMR 9189 CRIStAL, F-59000 Lille, France}
}
\email{mikael.monet@inria.fr}

\begin{abstract}
The Shapley value is a game-theoretic notion for wealth distribution that is nowadays extensively used to explain complex data-intensive computation, for instance, in network analysis or machine learning.
Recent theoretical works show that query evaluation over relational databases fits well in this explanation paradigm.
Yet, these works fall short of providing practical solutions to the computational challenge inherent to the Shapley computation. 
We present in this paper two practically effective solutions for computing Shapley values in query answering.
We start by establishing a tight theoretical connection to the extensively studied problem of query evaluation over probabilistic databases, which allows us to obtain a polynomial-time algorithm for the class of queries for which probability computation is tractable.
We then propose a first practical solution for computing Shapley values that adopts tools from probabilistic query evaluation. In particular, we capture the dependence of query answers on input database facts using Boolean expressions (data provenance), and then transform it, via Knowledge Compilation, into a particular circuit form for which we devise an algorithm for computing the Shapley values. 
Our second practical solution is a faster yet inexact approach that transforms the provenance to a Conjunctive Normal Form and uses a heuristic to compute the Shapley values.
Our experiments on TPC-H and IMDB demonstrate the practical effectiveness of our solutions.       

\end{abstract}

\maketitle
\pagestyle{plain}

\def\e#1{\emph{#1}}

\section{Introduction}
\label{sec:intro}
Explaining query answers has been the objective of extensive research in recent
years~\cite{DBLP:conf/tapp/SalimiBSB16,DBLP:conf/pods/GreenT17,DBLP:journals/pvldb/RoyOS15,DBLP:conf/icdt/LivshitsBKS20,deutch2020explaining}.
A prominent approach is to devise an explanation based on the facts that were
used for deriving the answers; these facts are often termed \e{provenance} or
\e{lineage} of the query answer.  For illustration, consider a query asking
whether there exists a route from the USA to France with at most one connection
over a database of airports and flights. The answer is Boolean, and upon
receiving a positive answer one may seek explanations of why it is so. The
basis for such explanations would include all details of qualifying routes that
are used in the derivation of the answer. 
Unfortunately, the number of relevant routes  might be huge (specifically,
quadratic in the database size). Moreover, it is conceivable that
different facts differ considerably in their importance to the answer at hand;
for instance, some flights may be crucial to enabling the USA-France
connection, while others may be easily replaced by alternatives.

To address these issues, there have been several proposals for principled ways
of \e {quantifying} the contribution of input facts to query
answers~\cite{DBLP:journals/pvldb/MeliouGMS11,DBLP:conf/tapp/SalimiBSB16,DBLP:conf/icdt/LivshitsBKS20,DBLP:journals/pvldb/MeliouRS14}.
We focus here on the recent approach of Livshits et
al.~\cite{DBLP:conf/icdt/LivshitsBKS20} that applies to this setting the notion of \e{Shapley
values}~\cite{shapley1953value}---a game-theoretic function for distributing the
wealth of a team in a cooperative game. This function has strong theoretical
justifications~\cite{1988TSv}, and indeed, it has been applied across various
fields such as economics, law,  environmental science, and network analysis. It
has also been used for explanations in data-centric paradigms such
as knowledge
representation~\cite{DBLP:journals/ai/HunterK10,DBLP:conf/ijcai/YunVCB18} and
machine learning~\cite{lundberg2017unified,lundberg2020local}.  In the context of
relational databases, given a query~$q(\bar x)$, a database~$D$, an input fact~$f\in D$
and a tuple~$\bar{t}$ of same arity as~$\bar x$, \emph{the Shapley value of~$f$ in~$D$ for query~$q(\bar x)$
and tuple~$\bar{t}$} intuitively represents the contribution of~$f$ to the
presence (or absence) of~$\bar{t}$ in the query result.

Livshits et al.~\cite{DBLP:conf/icdt/LivshitsBKS20,reshef2020impact} 
initiated the study of
 the computational complexity of calculating Shapley values in query answering. They showed mainly lower bounds on the complexity of the problem, with the exception of the sub-class of self-join free SPJ queries called \e{hierarchical}, where they gave a polynomial-time algorithm. The results are more positive if imprecision is allowed, as they showed that the problem admits a tractable approximation scheme (FPRAS, to be precise) via Monte Carlo sampling.
The state of affairs is that the class of known tractable cases (namely the hierarchical conjunctive queries) is highly restricted, and the approximation algorithms with theoretical guarantees are impractical in the sense that they require a large number of executions of the query over database subsets (the samples). 
Hence, the theoretical analysis of Livshits et al.~\cite{DBLP:conf/icdt/LivshitsBKS20,reshef2020impact} does not provide sufficient evidence of practical feasibility for adopting the Shapley value as a measure of responsibility in query answering. 
Moreover, 
the results of Livshits et al.~\cite{DBLP:conf/icdt/LivshitsBKS20} imply that, for self-join
free SPJ queries the class of tractable queries for computing Shapley values
coincides with the class of tractable queries in probabilistic
tuple-independent databases~\cite{dalvi2007efficient}. Yet, no direct
connection has been made between these two problems and, theoretically speaking, it
has been left unknown whether algorithms for probabilistic databases can be
used for Shapley computation.  

\revision{Recently, Van den Broeck et
al.~\cite{van2021tractability} and Arenas et
al.~\cite{arenas2021tractability,arenas2021complexity}  investigated the
computational complexity of the \emph{SHAP-score}~\cite{lundberg2017unified}, a
notion used in machine learning for explaining the predictions of a  model. 
While both are based on the general notion of Shapley value,
the SHAP-score for machine learning and
Shapley values for databases are different.
In the latter case, the players are the tuples of the database and the game function that is
used is simply the value of the query on a subset of the database, while in the former case, the players are the features of the model and
the game function is a conditional expectation of the model's output (see Section~\ref{sec:exp_approx} for a more formal definition of the SHAP-score).
Remarkably, Van den Broeck et
al.~\cite{van2021tractability} have shown that computing the
SHAP-score is equivalent (in terms of polynomial-time reductions) to the
problem of computing the expected value of the model.}
\revision{
 One of our contributions is to
show that the techniques developed
by~\cite{van2021tractability,arenas2021tractability,arenas2021complexity} can
be adapted to the context of Shapley values for databases. For instance, by adapting
to our context the proof of Van den Broeck et
al.~\cite{van2021tractability} that computing the SHAP-score reduces
to computing the expected value of the model, we resolve the aforementioned open
question affirmatively: we prove that Shapley computation can be efficiently
(polynomial-time) reduced to probabilistic query answering. Importantly, this
applies not only to the restricted class of SPJ queries without self-joins, but
to \emph{every} database query.  Hence, extending theory to practice, one can
compute the Shapley values using a query engine for probabilistic databases.
}

In turn, a common approach that was shown to be practically effective for
probabilistic databases is based on \e{Knowledge
Compilation}~\cite{darwiche2002knowledge,jha2013knowledge}.  In a nutshell, the idea is to
first compute the Boolean provenance of a given output tuple in the sense of Imielinski and
Lipski~\cite{imielinski84Incomplete},
and then to “compile” the provenance
into a particular circuit form that is more favorable for probability
computation.
Specifically, the target class of this compilation is that of \e{deterministic
and decomposable circuits} (d-D).
In our case, rather than going through probabilistic databases, we devise a
more efficient approach that computes the Shapley values directly from the d-D
circuit.  \revision{This is similar to how Arenas at al.~\cite{arenas2021tractability,arenas2021complexity} directly prove
that the SHAP-score can be computed efficiently over such circuits, without using the more general results of~\cite{van2021tractability}.
By adapting the proof of~\cite{arenas2021tractability,arenas2021complexity},} we show how, given a d-D circuit
representing the provenance of an output tuple, we can efficiently compute
the Shapley value of every input fact. While the aforementioned
properties of the circuit are not guaranteed in general (beyond the class of
hierarchical queries), we empirically show the applicability and usefulness of
the approach even for non-hierarchical queries.

\revision{Our experimental results (see below) indicate that our exact computation algorithm is fast in most cases, but is too costly in others. For the latter cases, we propose a heuristic approach to
retrieve the relative order of the facts by their Shapley values, without actually computing these values. Indeed, determining the most influential facts is in many cases already highly useful, even if their precise contribution remains unknown. The solution that we propose to this end is termed \e{\cprox{}}; it is based on a transformation of the provenance to Conjunctive Normal Form (CNF) and using it to compute proxy values intuitively based on (1) the number of clauses in which a variable occur and (2) its alternatives in each clause. These are two aspects that are correlated with Shapley values. The proxy values may be very different from the real Shapley values, and yet, when we order facts according to their proxy values we may intuitively get an ordering that is similar to the order via Shapley. Our experiments validate that this intuition indeed holds for examined benchmarks.  }

We have experimented with multiple queries from the two standard benchmarks TPC-H and IMDB. Our main findings are as follows. \revision{In most cases (98.67\% of the IMDB output tuples and 83.83\% for TPC-H), our exact computation algorithm terminates in 2.5 seconds or less, given the provenance expression. In the vast majority of remaining cases the execution is very costly, typically running out of memory already in the Knowledge Compilation step. By contrast, our inexact solution \e{\cprox{}} is extremely fast even for these hard cases -- it typically terminates in a few milliseconds with the worst observed case (an outlier) being 4 seconds. In fact, it is faster by several orders of magnitude than sampling-based approximation techniques (the \mc{} sampling proposed in \cite{DBLP:conf/icdt/LivshitsBKS20} as well as a popular sampling-based solution for Shapley values in Machine Learning (\kshap{}~\cite{lundberg2017unified}). To measure quality, we use \e{\cprox{}} to rank the input tuples, and compared the obtained ranked lists to ranking by actual Shapley values (in cases where exact computation has succeeded), using the standard measures of nDCG and Precision@k. Our solution outperforms the competitors in terms of quality as well. }

\revision{
We then propose a simple hybrid approach: execute the exact algorithm until it either terminates or a timeout elapses. If we have reached the timeout, resort to executing \e{\cprox{}} and rank the facts based on the obtained values. We show experiments with different timeout values, justifying our choice of 2.5 seconds.}

Hence, our contributions are both of a theoretical and practical nature and can be summarized as follows.
\begin{compactitem}
    \item \revision{By adapting the proof technique of~\cite{van2021tractability},} we establish a fundamental result about the complexity of computing Shapley values over relational queries: Shapley values can be computed in polynomial time--in data complexity--whenever the query can be evaluated in polynomial time over tuple-independent probabilistic databases (Proposition~\ref{prp:to-pqe}). This holds for \e{every} query.
    \item \revision{By adapting the proof technique of~\cite{arenas2021tractability,arenas2021complexity},} we devise a novel algorithm for computing Shapley values for query evaluation via compilation to a deterministic and decomposable circuit (Proposition~\ref{prp:shap-via-kc}). We show that this algorithm is practical and
    has the theoretical guarantee of running in polynomial time in the size of the circuit. 
    \item We present a novel heuristic, \cprox{}, that is fast yet inexact, and is practically effective if we are interested in ranking input facts by their contribution rather than computing exact Shapley values (Section~\ref{sec:approx}). 
    \item We describe a thorough experimental study of our algorithms over realistic data and show their efficiency (Section~\ref{sec:experiments}).

\end{compactitem}

\paragraph{Related work.} 
Existing models for explaining database query results may roughly be divided in
two categories: (1) models that are geared for tracking/presenting provenance
of output tuples, e.g., the set of all input facts participating in their
computation \cite{cuiweinerwidom}, possibly alongside a description of the ways
they were used, in different granularity levels
(e.g., \cite{green2007provenance,whyprovenance,Buneman2008}); (2) models that quantify
contributions of input facts
\cite{DBLP:journals/pvldb/MeliouGMS11,DBLP:conf/tapp/SalimiBSB16,DBLP:conf/icdt/LivshitsBKS20,DBLP:journals/pvldb/MeliouRS14},
which is the approach that we follow here. Works in the latter context often
have connections with the influential line of work on probabilistic
databases~\cite{suciu2011probabilistic}, and we show that this is the case for
Shapley computation as well.

As already mentioned, an important point of comparison is the work of Van
den Broeck et al.~\cite{van2021tractability} and that of Arenas et
al.~\cite{arenas2021tractability,arenas2021complexity} on the
\emph{SHAP-score}. \revision{While we show that the proof techniques developed
in this area can be adapted to the context of relational databases, we point
out that the two sets of results obtained (for SHAP-score and for Shapley
values for databases) seem incomparable, as we do not see a way of
proving results for Shapley value for query answering using the results on the
SHAP-score, or vice-versa. In fact, this adaptation only works up to a certain
point. For instance, the \emph{efficiency axiom} of the Shapley
value immediately implies that  computing the expected value of a model 
can be reduced in polynomial time to computing the SHAP-score of its features;
in contrast, this axiom does not seem to yield any clear such implication in
our context (see our Open Problem~\ref{open:pqe-to-shap} and the discussion
around it).}

\paragraph{Paper organization.} 
We formalize the notion of Shapley values for query answering in
Section~\ref{sec:def}. In Section \ref{sec:pqe} we present the theoretical
connection to probabilistic databases and its implications. Our exact
computation algorithm is presented in Section \ref{sec:lineages} and our
heuristic in Section \ref{sec:approx}.  Experimental results are presented in
Section \ref{sec:experiments} and we conclude in Section \ref{sec:conc}.

\begin{figure*}
    \centering
    % \boxed{
    \begin{minipage}[c]{.39\textwidth}
    \subfloat[Database of flights and airports]{
    \begin{tabular}{c|c|c|}
    \cline{2-3}
    & \multicolumn{2}{c|}{\cellcolor[HTML]{C0C0C0}\textsc{Flights} (endo)}\\
    \cline{2-3}
    & \textbf{Src} & \textbf{Dest} \\
    \hline
    $a_1$ & JFK & CDG\\
    $a_2$ & EWR & LHR\\
    $a_3$ & BOS & LHR\\
    $a_4$ & LHR & CDG\\
    $a_5$ & LHR & ORY\\
    $a_6$ & LAX & MUC\\
    $a_7$ & MUC & ORY\\
    $a_8$ & LHR & MUC\\    
    \hline
    \end{tabular}
    \quad
    \begin{tabular}{c|c|c|}
    \cline{2-3}
    & \multicolumn{2}{c|}{\cellcolor[HTML]{C0C0C0}\textsc{Airports} (exo)}\\
    \cline{2-3}
    & \textbf{Name} & \textbf{Country} \\
    \hline
    $b_1$ & JFK & USA\\
    $b_2$ & EWR & USA\\
    $b_3$ & BOS & USA\\
    $b_4$ & LAX & USA\\
    $b_5$ & LHR & EN\\
    $b_6$ & MUC & GR\\
    $b_7$ & ORY & FR\\
    $b_8$ & CDG & FR\\    
    \hline
    \end{tabular}
    \label{fig:running_database}
    }
    \end{minipage}%
    \begin{minipage}[c]{.29\textwidth}
    \centering
    \subfloat[Flights in graph view. Dark and light gray depict ``USA'' and ``FR'' airports respectively]{
    \label{fig:running_graph}
    \begin{tikzpicture}
    \begin{scope}[every node/.style={rectangle, rounded corners,thick,draw}]
        \node[fill=gray!30] (1) at (0,0) {\scriptsize JFK};
        \node[fill=gray!30] (2) at (0,1) {\scriptsize EWR};
        \node[fill=gray!30] (3) at (0,2.5) {\scriptsize BOS};
        \node[fill=gray!30] (4) at (0,3.5) {\scriptsize LAX};
        \node (5) at (1.25, 1.75) {\scriptsize LHR};
        \node (6) at (1.25,3.5) {\scriptsize MUC};
        \node[fill=gray!10] (7) at (2.5,2.5) {\scriptsize ORY};
        \node[fill=gray!10] (8) at (2.5,0) {\scriptsize CDG};      
    \end{scope}
    \begin{scope}[every edge/.style={draw=black, thick}]
        \path [->] (1) edge node[above] {$a_1$} (8);
        \path [->] (2) edge node[above, sloped] {$a_2$} (5);
        \path [->] (3) edge node[above, sloped] {$a_3$} (5);
        \path [->] (5) edge node[above, sloped] {$a_5$} (7);
        \path [->] (5) edge node[above, sloped] {$a_4$} (8);
        \path [->] (4) edge node[above] {$a_6$} (6);
        \path [->] (6) edge node[above, sloped] {$a_7$} (7);
        \path [->] (5) edge node[right] {$a_8$} (6);
    \end{scope}
    \end{tikzpicture}    
    }
    \end{minipage}%
    \begin{minipage}[c]{.28\textwidth}
    \centering
    \subfloat[$q$ is a Boolean union of conjunctive queries (UCQ)]{
    \label{fig:query_q}
    \begin{tabular}{l}
        $\begin{aligned} 
            &\texttt{Routes from ``USA'' to ``FR'' with} \\
            &\texttt{one or less connecting flights}
        \end{aligned}$\\
        \hline\\[.1em]
        $\begin{aligned} 
        q_1\, =\, &\exists x,y:\, \textsc{Airports}(x, \text{``USA''})\, \land \\ &\textsc{Airports}(y, \text{``FR''}) \, \land \\ &\textsc{Flights}(x, y)\\
        q_2\, =\, &\exists x,y,z:\,\textsc{Airports}(x, \text{``USA''})\, \land \\ &\textsc{Airports}(z, \text{``FR''})\, \land \\ &\textsc{Flights}(x, y)\, \land \, \textsc{Flights}(y, z) \\
        q\, =\, &q_1 \lor q_2
        \end{aligned}$
    \end{tabular}
    }
    \end{minipage}%    
    \vspace{.3cm}
    \begin{minipage}[c]{\textwidth}
    \centering
    \subfloat[The lineage $\lin(q, D)$ as a formula in Disjunctive Normal Form (DNF)]{
    \label{fig:provenance_q}
    $\begin{aligned} 
    \underbrace{\left(a_1 \wedge b_1 \wedge b_8 \right)}_{q_1 \text{ provenance}} \vee 
    \underbrace{\left(a_2 \wedge a_4 \wedge b_2 \wedge b_8\right) \vee    
    \left(a_2 \wedge a_5 \wedge b_2 \wedge b_7\right) \vee 
    \left(a_3 \wedge a_4 \wedge b_3 \wedge b_8 \right) \vee
    \left(a_3 \wedge a_5 \wedge b_3 \wedge b_7 \right) \vee
    \left(a_6 \wedge a_7 \wedge b_5 \wedge b_7 \right)}_{q_2 \text{ provenance}}
    \end{aligned}$
    }%
    \caption{The database, query $q$ and its lineage used for our running example}
    \label{fig:airports_example}
    \end{minipage}%
\end{figure*}

\section{The Shapley value of facts}
\label{sec:def}
We define here the main notion and illustrate it with an example.\\

\paratitle{Relational databases and queries.}
Let~$\Sigma = \{R_1,\ldots,R_n\}$ be a \emph{signature}, consisting of
\emph{relation names}~$R_i$ each with its associated \emph{arity}~$\ar(R_i) \in
\mathbb{N}$, and~$\const$ be a set of \emph{constants}.  A \emph{fact} over
$(\Sigma,\const)$ is simply a term of the form~$R(a_1,\ldots,a_{\ar(R)})$,
for~$R \in \Sigma$ and~$a_i \in \const$.
A~\emph{($\Sigma,\const)$-database}~$D$, or simply a \emph{database}~$D$, is a
finite set
of facts over $(\Sigma,\const)$.
We assume familiarity with the most common
classes of query languages and refer the reader
to~\cite{abiteboul1995foundations} for the basic definitions.  In particular, we recall the
equivalence between relational algebra and relational
calculus~\cite{abiteboul1995foundations}, and the fact that
Select-Project-Join-Union (SPJU) queries are equivalent to unions of conjunctive queries (UCQs).
Depending of the context and for consistency with relevant past publications,
we will use terminology of either relational calculus or relational algebra.
What we call a \emph{Boolean query} is a query~$q$ that takes as input a
database~$D$ and outputs~$q(D) \in \{0,1\}$.  If~$q(\bar{x})$ is a query with
free variables~$\bar{x}$ and~$\bar{t}$ is a tuple of constants of same length
as~$\bar{x}$, we denote by~$q[\bar{x}/\bar{t}]$ the Boolean query defined by:
$q[\bar{x}/\bar{t}](D) = 1$ if and only if~$\bar{t}$ is in the output
of~$q(\bar{x})$ on~$D$. \\

\paratitle{Shapley values of facts.}
Following~\cite{DBLP:conf/icdt/LivshitsBKS20,reshef2020impact}, we use the notion of Shapley values~\cite{shapley1953value} to attribute a contribution to facts of an
input database. In this context, the database~$D$ is traditionally partitioned into two sets of facts: a set~$\dx$ of so-called \emph{exogenous} facts,
and a set~$\dn$ of \emph{endogenous} facts. The idea is that exogenous facts are considered as given,
while endogenous facts are those to which we would like to attribute contributions. Let~$q$ be a Boolean query
and~$f \in \dn$ be an endogenous fact. \emph{The Shapley value of~$f$ in~$D$
for query~$q$}, denoted~$\shap(q,\dn,\dx,f)$, is defined as

\begin{align}
\label{eq:defshap}
&\shap(q,\dn,\dx,f) \ \defeq \\ \notag
&\sum_{E\subseteq \dn \setminus \{f\}} \frac{|E|!(|\dn|-|E|-1)!}{|\dn|!} \big(
q(\dx \cup E\cup \{f\}) - q(\dx \cup E)\big).
\end{align}
Notice that here, $|E|! (|\dn|-|E|-1)!$ is the number of permutations of~$\dn$
with all endogenous facts in~$E$ appearing first, then~$f$, and finally, all
the other endogenous facts.  Intuitively then, the value $\shap(q,\dn,\dx,f)$
represents the contribution of~$f$ to the query's output: the higher this value
is, the more~$f$ helps in satisfying~$q$.

For non-Boolean queries $q(\bar{x})$, we are interested in the Shapley value of the fact $f$ for every individual tuple $\bar{t}$ in the output~\cite{DBLP:conf/icdt/LivshitsBKS20}. 
The extension to non-Boolean $q(\bar{x})$ is 
then straightforward: the Shapley value of the fact $f$ \emph{for the answer
$\bar{t}$ to $q(\bar x)$} is the value  $\shap(q[\bar{x}/\bar{t}],\dn,\dx,f)$. 
Therefore, the computational challenge reduces to that of the Boolean query $q[\bar{x}/\bar{t}]$.
Hence, in the theoretical analysis we focus on Boolean queries, and we go back to considering non-Boolean queries when we study the implementation aspects (starting in Section~\ref{subsec:strategy}).

\begin{example}
\label{expl:shapley}
Consider the database~$D$ and the Boolean query~$q$ from Figures
\ref{fig:running_database} and \ref{fig:query_q}. All facts in table \textsc{Flights} are endogenous, while facts in \textsc{Airports} are exogenous. To alleviate the notation we write, e.g., $a_1$ for $\textsc{Flights}(\textsc{JFK},\textsc{CDG})$. The query $q$ checks if there are routes from ``USA'' to ``FR'' with one or less connecting flights. 
Let us compute the Shapley value of all endogenous facts. First, we notice that fact $a_8$ is not part of any valid route, so $\shap(q,\dn,\dx,a_8)=0$ by Equation~\eqref{eq:defshap}. Next, let us focus on~$a_1$. Since $a_1$ is a valid route on its own, adding it to any subset of (endogenous) facts $E$ such that $E$ does not contain a valid route results in $q(\dx \cup E\cup \{a_1\}) - q(\dx \cup E) = 1$ (for all other subsets the difference will be 0). The relevant subsets are the empty set, all singletons $\{a_i\}$ for $2 \leq i \leq 8$ (7 singletons), all the pairs of tuples from $a_2, \ldots, a_8$ excluding the pairs $\{a_2, a_4\}$, $\{a_2, a_5\}$, $\{a_3, a_4\}$, $\{a_3, a_5\}$, and $\{a_6, a_7\}$ (so $\binom{7}{2} - 5 = 16$ pairs), 
the quadruples $\{a_2, a_3, a_6, a_8\}$, $\{a_2, a_3, a_7, a_8\}$, $\{a_4, a_5, a_6, a_8\}$, and $\{a_4, a_5, a_7, a_8\}$, and overall 14 triplets (left to the reader).
Summing it all up results in 
\begin{align*}
\shap(q,\dn,\dx,a_1) =\, & 1 \cdot \frac{0! \cdot 7!}{8!} + 7 \cdot \frac{1! \cdot 6!}{8!} + 16 \cdot \frac{2!5!}{8!} + \\ 
& 14 \cdot \frac{3!4!}{8!} + 4 \cdot \frac{4!3!}{8!} =  \frac{43}{105} \approx 0.4095.
\end{align*}
Similarly one can compute the Shapley value of the remaining facts, and find that for $a_i \in \{a_2, a_3, a_4, a_5\}$ it holds that\linebreak $\shap(q,\dn,\dx,a_i) = \frac{23}{210} \approx 0.1095$, and that for $a_i \in \{a_6, a_7\}$ we have $\shap(q,\dn,\dx,a_i) = \frac{8}{105} \approx 0.0762$.
\end{example}

\section{Reduction to probabilistic databases}
\label{sec:pqe}

In this section we investigate the complexity of computing Shapley values. As explained in the previous section, the non-Boolean setting of the problem may be reduced to that of Boolean queries, so we will study the following problem for a given Boolean query~$q$. 
\begin{center}
\fbox{\begin{tabular}{rp{6.1cm}}
\small{PROBLEM}: & $\shap(q)$ 
\\
{\small INPUT}: & A database $D = \dx \cup \dn$ and an endogenous fact~$f \in \dn$. 
\\
{\small OUTPUT}: & The value $\shap(q,\dn,\dx,f)$.
\end{tabular}}
\end{center}

Note that for the sake of complexity analysis, the query~$q$ is assumed to be fixed, so that each query gives
rise to a different computational problem; we are then considering what is called the \emph{data
complexity}~\cite{vardi1982complexity}. This assumption is motivated by the fact that
in practice, the queries are much smaller than the databases.

The complexity of this problem has been studied
in~\cite{DBLP:conf/icdt/LivshitsBKS20,reshef2020impact}, where in particular a dichotomy
has been obtained for \emph{self-join--free Boolean conjunctive queries}
(sjfbcqs).  There, the authors show that, for every sjfbcq~$q$, either~$q$ is
\emph{hierarchical} (we will not need to define this notion here) and~$\shap(q)$ can
be solved in polynomial time, or~$q$ is not hierarchical and then~$\shap(q)$ is
intractable (specifically,~$\fptoshp$-hard). It turns out that the tractability
criterion that is obtained---being hierarchical---is exactly the same as in the
context of probabilistic query evaluation (PQE); see,
e.g.,~\cite{dalvi2007efficient,dalvi2013dichotomy}.  In fact,
the main result of this section is that this is not a coincidence: we prove that,
for \emph{every} 
Boolean query~$q$ (not just for sjfbcqs), if PQE is tractable for~$q$ then so is the
problem~$\shap(q)$. Since PQE has been intensively studied already,
our result allows us to vastly extend the tractable
cases identified in~\cite{DBLP:conf/icdt/LivshitsBKS20,reshef2020impact}.
We now proceed with the definitions and proof of this result, and
explain its consequences.\\

\paratitle{Probabilistic query evaluation.}
A \emph{tuple-independent (TID)
database} is a pair consisting of a database $D$ and a function $\pi$
mapping each fact $f \in D$ to a probability value $\pi(f) \in [0,1]$.
The TID $(D, \pi)$ defines a probability distribution~$\Pr_\pi$ on $D' \subseteq D$, where
$\Pr_\pi(D') \defeq \prod_{f \in D'} \pi(f) \times \prod_{f \in D \backslash D'} (1 - \pi(f))$.
Given a Boolean query $q$, \emph{the probability that $q$ is satisfied by~$(D,\pi)$} is 
$\Pr(q, (D, \pi)) \defeq \sum_{D' \subseteq D\text{~s.t.~}q(D')=1} \Pr_\pi(D')$.
The \emph{probabilistic query evaluation problem for $q$}, $\pqe(q)$ for short, is then defined as follows.

\begin{center}
\fbox{\begin{tabular}{rp{6.1cm}}
\small{PROBLEM}: & $\pqe(q)$ 
\\
{\small INPUT}: & A tuple-independent database $(D,\pi)$. 
\\
{\small OUTPUT}: & The value $\Pr(q, (D, \pi))$.
\end{tabular}}
\end{center}

For two computational problems~$A$ and~$B$, we write~$A \tr B$ to assert the existence of a polynomial-time Turing reduction
from~$A$ to~$B$. We are ready to state the main result of this section.

\begin{proposition}
\label{prp:to-pqe}
For every Boolean query~$q$, we have that\linebreak
 $\shap(q) \tr \pqe(q)$.
\end{proposition}

This result implies that for any query~$q$ for which~$\pqe(q)$ is tractable then so
is~$\shap(q)$.  \revision{Dalvi and Suciu~\cite{dalvi2013dichotomy} showed a
dichotomy for unions of conjunctive queries:
for every such query~$q$, 
either $\pqe(q)$ is solvable in polynomial time, in which case~$q$ is called \emph{safe}\footnote{\revision{This notion of safety is distinct from the “usual” notion of query safety~\cite{abiteboul1995foundations} that ensures domain independence.}}, or 
$\pqe(q)$ is $\fptoshp$-hard (and~$q$ is called \emph{unsafe}).} Therefore, we obtain as a direct corollary of
Proposition~\ref{prp:to-pqe} that~$\shap(q)$ can be solved in polynomial time
for all safe queries.

\begin{corollary}
\label{co:safe}
If~$q$ is a safe UCQ then~$\shap(q)$ can be solved in polynomial time.
\end{corollary}
In particular, this corollary generalizes the tractability result 
obtained in~\cite{DBLP:conf/icdt/LivshitsBKS20}, to account for CQs with self-joins and even unions of such queries.
We now prove
Proposition~\ref{prp:to-pqe}. 
\begin{proof}[Proof of Proposition~\ref{prp:to-pqe}]
For a Boolean query~$q$, database~$D = \dx \cup \dn$,
and integer~$k \in \{0,\ldots,|\dn|\}$,
define
\[\shslices(q,\dx,\dn,k) \ \defeq \  |\{E \subseteq \dn \mid |E|=k \text{ and } q(\dx \cup E)=1 \}|.\]
Then, by grouping by size the terms~$E$ from Equation~\ref{eq:defshap} we obtain 
\begin{align}
\label{eq:kc-slices}
\shap(q,\dn,\dx,f) \  & = 
  \sum_{k=0}^{|\dn|-1} \frac{k!(|\dn|-k-1)}{|\dn|}\bigg( \\
 \nonumber&\hspace{25pt}  \shslices(q,\, \dx \cup \{f\},\, \dn \setminus \{f\},\, k)\\
 \nonumber&\hspace{25pt}  - \shslices(q,\, \dx,\, \dn \setminus \{f\},\, k)\bigg).
\end{align}
All arithmetical terms (such as~$k!$ or~$|\dn|!$) can be computed
in polynomial time. Therefore, to prove that~$\shap(q) \tr \pqe(q)$, it is
enough to show that, given an oracle to the problem~$\pqe(q)$, we can compute
in polynomial time the quantities $\shslices(q,\dx,\dn,k)$, for some
arbitrary~$D=\dx \cup \dn$ and $k \in \{0,\ldots,|\dn|\}$. This is what we do next.

\revision{The proof is similar to
that of~\cite[Theorem 2]{van2021tractability} in the context of SHAP-score for
machine learning (but, as explained in the Introduction, the two results seem to be
incomparable).}

We wish to compute~$\shslices(q,\dx,\dn,k)$, for some database
$D=\dx \cup \dn$ and integer~$k \in \{0,\ldots |\dn|\}$.
Let~$n=|\dn|$ be the number of endogenous facts of~$D$.
For~$z \in \mathbb{Q}$, we define a TID database~$(D_z,\pi_z)$ as follows: $D_z$ contains all the facts of~$D$, and for an exogenous fact~$f$ of~$D$ we define~$\pi_z(f) \defeq 1$
while for an endogenous fact~$f$ of~$D$ we define ~$\pi_z(f) \defeq \frac{z}{1+z}$.
It is then routine to show that the following relation holds:
\[(1+z)^n \Pr(q,(D_z,\pi_z)) = \sum_{i=0}^n z^i\, \shslices(q,\dx,\dn,i).\]
This suffices to conclude the proof. Indeed, we now call an oracle
to~$\pqe(q)$ on~$n+1$ databases~$D_{z_0},\ldots,D_{z_n}$ for~$n+1$ arbitrary distinct
values~$z_0,\ldots,z_n$, forming a system of linear equations as given by the
relation above. Since the corresponding matrix is a Vandermonde with distinct coefficients, it
is invertible, so we can compute in polynomial time the value~$\shslices(q,\dx,\dn,k)$.
\end{proof}

\revision{A intriguing natural question is whether the converse of Proposition~\ref{prp:to-pqe} is true, that is, whether we also have $\pqe(q) \tr \shap(q)$.
This is true when $q$ is a self-join--free conjunctive query: indeed, by the results of~\cite{DBLP:conf/icdt/LivshitsBKS20}, either~$q$ is hierarchical and then both~$\pqe(q)$ and~$\shap(q)$
can be solved in polynomial time (hence $\pqe(q) \tr \shap(q)$), or~$q$ is not hierarchical and then~$\shap(q)$ is $\fptoshp$-hard (which, together with the fact that~$\pqe(q) \in \fptoshp$, implies $\pqe(q) \tr \shap(q)$).
However, to the best of our knowledge, the existence of such a reduction in the general case is unknown.
\begin{openpb}
	\label{open:pqe-to-shap}
	Do we have $\pqe(q) \tr \shap(q)$ for every Boolean query~$q$?
\end{openpb}
Interestingly, we note that this direction is trivial in the setting of SHAP-scores. Indeed, this is directly
implied by the \emph{efficiency axiom} of the Shapley value; see, e.g., \cite[Lemma 4.2]{arenas2021complexity} or \cite[Equation 5]{van2021tractability}.
In our case, this axiom only gives us the following equality:
\[\sum_{f\in \dn} \shap(q,\dn,\dx,f)\, = \, q(\dn \cup \dx) - q(\dx),\]
which does not seem to help in showing~$\pqe(q) \tr \shap(q)$.
}

\section{Exact Computation Through Knowledge Compilation}
\label{sec:lineages}

Motivated by the connection to PQE that we have seen in Section~\ref{sec:pqe}, we now investigate whether an approach using knowledge compilation can be used for computing Shapley values. Indeed, a common method to compute the probability that a probabilistic
database~$(D,\pi)$ satisfies a Boolean query~$q$ is to first compute the
\emph{lineage} of~$q$ on~$D$ in a formalism from knowledge compilation, and then to use the good
properties of said formalism to compute~$\Pr(q,(D,\pi))$ in linear
time~\cite{suciu2011probabilistic,jha2013knowledge,monet2020solving}.  \revision{Recently, Arenas and others~\cite{arenas2021tractability,arenas2021complexity} showed that this
approach is also viable for the notion of SHAP-score used in machine
learning, by proving that SHAP-scores can be computed in polynomial time when
the models are given as circuits from knowledge compilation.  By reusing some
of these techniques,
we can show that this method can also be used in our setting
for computing Shapley values of database facts. Again, to the best of our knowledge, the two results are incomparable, i.e.,
we are not aware of a reduction in either direction between the two problems.}
We start by formally defining
the notions of lineage and the relevant circuit classes from knowledge compilation.

\paragraph{Boolean functions and query lineages.}
Let~$X$ be a finite set of variables.  An \emph{assignment~$\nu$ of~$X$} is a
subset~$\nu \subseteq X$ of~$X$. We denote by~$2^X$ the set of all assignments
of~$X$.  A \emph{Boolean function}~$\phi$ over~$X$ is a function~$\phi:2^X\to
\{0,1\}$.
An assignment~$\nu \subseteq X$ is \emph{satisfying} if~$\phi(\nu) =1$.  We denote
by~$\sat(\phi)\subseteq 2^X$ the set of all satisfying assignments of~$\phi$,
and by ~$\ssat(\phi)$ the size of this set.  For~$k\in \mathbb{N}$, we
define~$\sat_k(\phi) \defeq \sat(\phi) \cap \{\nu \subseteq X \mid |\nu|=k\}$, that
is, the set of satisfying assignments of~$\phi$ of Hamming weight~$k$, and
let~$\ssat_k(\phi)$ be the size of this set.

Let~$q$ be a Boolean query and~$D$ be a database.  The \emph{lineage $\lin(q,
D)$} is the (unique) Boolean function whose variables are the facts of~$D$, and
that maps each sub-database $D' \subseteq D$  to~$q(D')$.  
This
definition extends straightforwardly to queries with free variables as follows:
if~$q(\bar{x})$ is a query with free variables~$\bar{x}$ and~$\bar{t}$ is a
tuple of constants of the appropriate size, then~$\lin(q[\bar{x}/\bar{t}]),D)$
is the \emph{lineage for the tuple~$\bar{t}$}.

\begin{example}
\label{expl:lineage_ex}
Consider again the database~$D$ and the Boolean query $q$ from Figures
\ref{fig:running_database} and \ref{fig:query_q}. 
In Figure \ref{fig:provenance_q}, the lineage
$\lin(q,D)$ is represented as a formula in
\emph{disjunctive normal form} (DNF). 
\end{example}

\begin{figure}
\begin{scalebox}{0.9}{\input{figures/circuit.pspdftex}}\end{scalebox}
\caption{
Deterministic and decomposable circuit for~$\elin(q,\dx,\dn)$ from the running example.
}
\label{fig:d-D}
\end{figure}

For our purposes, we will use a refinement of this lineage
that accounts for the nature of exogenous tuples; specifically, these tuples should be considered as always
being part of the database. Let~$D=\dn \cup \dx$ be a database with endogenous tuples~$\dn$
and exogenous tuples~$\dx$, and let~$q$ be a Boolean query. Then the \emph{endogenous
lineage~$\elin(q,\dx,\dn)$} is the (unique) Boolean function whose variables are~$\dn$
and that maps every set~$E$ of endogenous facts to~$q(\dx \cup E)$. In other words,
$\elin(q,\dx,\dn)$ can be obtained from $\lin(q,D)$ by fixing all variables in~$\dx$
to the value~$1$. Again, we extend this definition to queries with free variables
by using the function~$\elin(q[\bar{x}/\bar{t}],\dx,\dn)$.

\begin{example}
\label{expl:lineage}
Continuing the previous example, 
the endogenous lineage $\elin(q,\dn,\dx)$ can be represented as a DNF by
    \begin{align*} 
    a_1 \vee 
    \left(a_2 \wedge a_4\right) \vee    
    \left(a_2 \wedge a_5 \right) \vee 
    \left(a_3 \wedge a_4 \right) \vee
    \left(a_3 \wedge a_5 \right) \vee
    \left(a_6 \wedge a_7 \right).
    \end{align*}
\end{example}

In the last two examples, lineages were represented with Boolean formulas in DNF.
Since a lineage is a Boolean function, it can be represented with
any formalism that allows to represent Boolean functions. We next review some
classes of circuits from the field of knowledge compilation that will be relevant
for our work.

\paragraph{Knowledge compilation classes.}
Let~$C$ be a Boolean circuit, featuring~$\land$,~$\lor$,~$\lnot$, and
variable gates,
with the usual semantics.\footnote{We allow
unbounded-fanin~$\land$- and~$\lor$-gates, and also allow constant~$1$-gates
and constant~$0$-gates as, respectively, $\land$-gates with no input
and~$\lor$-gates with no inputs.} For a gate~$g$ of~$C$, we denote by $\vars(g)$  
the set of variables that have a directed path to~$g$.  An
$\land$-gate~$g$ of~$C$ is \emph{decomposable} if for every two input gates
$g_1\neq g_2$ of~$g$ we have $\vars(g_1) \cap \vars(g_2) = \emptyset$.  We
call~$C$ \emph{decomposable} if all~$\land$-gates are.  An $\lor$-gate~$g$
of~$C$ is \emph{deterministic} if the Boolean functions captured by each pair of distinct
input gates of~$g$ are pairwise disjoint; i.e.,
no assignment satisfies both.  We call~$C$
\emph{deterministic} if all~$\lor$-gates in it are.
A \emph{deterministic and decomposable (d-D~\cite{monet2020solving}) Boolean circuit} is a Boolean circuit that is
both deterministic and decomposable. If~$C$ is a Boolean circuit we write~$\vars(C)$ to denote
the set of variables that appear in it.

\begin{example}
Recall $\elin(q,\dn,\dx)$ from Example~\eqref{expl:lineage} represented as a
DNF.  Figure~\ref{fig:d-D} depicts a d-D circuit for
$\elin(q,\dn,\dx)$.  The output gate, for instance, is a deterministic~$\lor$-gate:
indeed, its left child requires~$a_1$ to be~$1$, whereas its right child
requires~$a_1$ to be~$0$. The right child of the output gate is a
decomposable~$\land$-gate: indeed, for its left child~$g_1$ we have~$\vars(g_1) =
\{a_1\}$, whereas for its right child~$g_2$ we have $\vars(g_2) =
\{a_2,a_3,a_4,a_5,a_6,a_7\}$, and these are indeed disjoint. The reader can easily
check that all other~$\lor$-gates are deterministic, and that all
other~$\land$-gates are decomposable.
\end{example}

\subsection{Algorithm}
The main result of this section
is then the following.

\begin{proposition}
\label{prp:shap-via-kc}
Given as input a deterministic and decomposable circuit~$C$
representing~$\elin(q,\dn,\dx)$ for a database~$D= \dx \cup \dn$ and Boolean
query~$q$, and an endogenous fact~$f\in \dn$, we can compute in polynomial time (in~$|C|$)
the value~$\shap(q,\dn,\dx,f)$.
\end{proposition}

Next, we prove Proposition~\ref{prp:shap-via-kc} 
and present the algorithm, and then explain in
Section~\ref{subsec:strategy} the architecture of the implementation.

\paragraph{Proof of Proposition~\ref{prp:shap-via-kc}}
\label{subsec:d-D-proof}

Let~$C$ be a deterministic and decomposable
circuit representing~$\elin(q,\dn,\dx)$, and let ~$f\in \dn$.  First, we complete
the circuit~$C$ so that all variables of~$\dn$ appear in~$C$.  Indeed, it could
be the case that~$\vars(C) \subsetneq \dn$: this happens for instance with
the deterministic and decomposable circuit in Figure~\ref{fig:d-D}, where the
endogenous fact~$a_8$ does not appear in the circuit. To do this, we
conjunct~$C$ with the conjunction~$\bigwedge_{f' \in \dn \setminus \vars(C)}
(f' \lor \lnot f')$.  Note that this does not change the semantics of the
circuit (as this conjunction always evaluates to~$1$) and that the resulting circuit is still
deterministic and decomposable. 
Now, let~$C_1$ (resp.,~$C_2$) be the Boolean circuit obtained from~$C$
by replacing all variable gates corresponding to the fact~$f$ by a
constant~$1$-gate (resp., by a constant~$0$-gate). Observe then that the
variables of~$C_1$ and~$C_2$ are exactly~$\dn \setminus \{f\}$, and moreover
that~$C_1$ and~$C_2$ are still deterministic and decomposable.
By definition of the endogenous lineage, we can
rewrite Equation~\eqref{eq:kc-slices} 
into the following.

\begin{align}
\label{eq:shap-ssat}
\shap(q,\dn,\dx,f) \  & =  \sum_{k=0}^{|\dn|-1} \frac{k!(|\dn|-k-1)}{|\dn|}\big(\\
\nonumber &\hspace{25pt} \ssat_k(C_1) - \ssat_k(C_2)\big).
\end{align}
Proposition~\ref{prp:shap-via-kc} will thus directly follow from the next
lemma.

\begin{lemma}
\label{lem:ssat-kc}
Given as input a deterministic and decomposable Boolean circuit~$C$ and an
integer~$k\in \{0,\ldots, |\vars(C)|\}$, we can compute in polynomial time the
quantity~$\ssat_k(C)$.
\end{lemma}
\begin{proof}
Our proof is similar to that of~\cite[Section 3.2]{arenas2021tractability}.  Let~$X
\defeq \vars(C)$ and $n \defeq |X|$.  First of all, we preprocess~$C$ so that
the fanin of every~$\lor$- and~$\land$-gate is exactly~$0$ or~$2$; this can 
simply be done by rewriting every~$\land$-gate of fanin~$m >2$ with~$m-1$ $\land$-gates of
fanin~$2$ (same for~$\lor$-gates), and adding a constant gate of the appropriate
type to every~$\lor$- and~$\land$-gate of fan-in~$1$. We then compute, for every
gate~$g$ of~$C$, the set of variables~$\vars(g)$ upon which the value of~$g$
depends. For a gate~$g$ of~$C$, let us denote by~$\phi_g$ the Boolean
function over the variables~$\vars(g)$ that is represented by this gate.  For a
gate~$g$ and an integer~$\ell\in \{0,\ldots,|\vars(g)|\}$, we
define~$\alpha_g^\ell \defeq \ssat_\ell(\phi_g)$, i.e. the number
of assignments of size~$\ell$ to ~$\vars(g)$ that satisfy~$\phi_g$. We will show
how to compute all the values~$\alpha_g^\ell$ for every gate~$g$ of~$C$
and~$\ell\in \{0,\ldots,|\vars(g)|\}$ in polynomial time. This will conclude
the proof since, for the output gate~$g_\out$ of~$C$, we have
that~$\alpha_{g_\out}^k = \ssat_k(f)$.
We will need the following notation: for two disjoint sets of
variables~$X_1,X_2$ and two subsets~$S_1 \subseteq 2^{X_1}$, $S_2 \subseteq
2^{X_2}$ of assignments to ~$X_1$ and~$X_2$, we denote by~$S_1 \otimes S_2 \subseteq 2^{X_1\cup X_2}$ the
set of assignments of~$X_1\cup X_2$ defined by~$S_1 \otimes S_2 \defeq
\{\nu_1\cup \nu_2 \mid \nu_1 \in S_2, \nu_2\in S_2\}$.  We next show how to
compute the values~$\alpha_g^\ell$ by bottom-up induction on~$C$.

\begin{description}
    \item[Variable gate.] If $g$ is a variable gate corresponding to some variable~$y$, then~$\vars(g)=\{y\}$.  Then, ~$\alpha_g^0$ is~$0$ and~$\alpha_g^1$ is~$1$.
    \item[$\lnot$-gate.] If $g$ is a~$\lnot$-gate with input gate~$g'$, then~$\alpha_g^\ell = \binom{|\vars(g)|}{l} - \alpha_{g'}^\ell$ for every~$\ell\in \{0,\ldots,|\vars(g)|\}$.

    \item[Deterministic $\lor$-gate.] If $g$ is a deterministic~$\lor$-gate
with no input then~$\phi_g$ is the Boolean function on variables~$\vars(g) =
\emptyset$ that is always false, hence~$\alpha_g^0=0$. Otherwise~$g$ has exactly two
input gates; let us denote them~$g_1$ and~$g_2$.
Observe that~$\vars(g) = \vars(g_1) \cup \vars(g_2)$ by definition.
Define~$S_1 \defeq \vars(g_2) \setminus \vars(g_1)$ and similarly~$S_2 \defeq
\vars(g_1) \setminus \vars(g_2)$.  Since ~$g$ is deterministic, we have:
\[\sat(\phi_g) = (\sat(\phi_{g_1}) \otimes 2^{S_1}) \cup (\sat(\phi_{g_2}) \otimes
2^{S_2})\] with the union being disjoint. By intersecting with the assignments
of~$\vars(g)$ of size~$\ell$, we obtain: \begin{align*} \sat_\ell(\phi_g) = &
\big[(\sat(\phi_{g_1}) \otimes 2^{S_1})\cap \{\nu\subseteq \vars(g) \mid |\nu|=\ell\}
\big]\\ & \cup \big[(\sat(\phi_{g_2}) \otimes 2^{S_2})\cap \{\nu\subseteq \vars(g)
\mid |\nu|=\ell\} \big] \end{align*} with again the middle union being disjoint,
therefore: \begin{align*} \ssat_\ell(\phi_g) = & |(\sat(\phi_{g_1}) \otimes 2^{S_1})\cap
\{\nu\subseteq \vars(g) \mid |\nu|=\ell\}|\\ & + |(\sat(\phi_{g_2}) \otimes
2^{S_2})\cap \{\nu\subseteq \vars(g) \mid |\nu|=\ell\}| \end{align*}

We now explain how to compute the first term, that is, $|(\sat(\phi_{g_1}) \otimes
2^{S_1})\cap \{\nu\subseteq \vars(g) \mid |\nu|=\ell\}|$;
the second term is similar.  This is equal\footnote{This comes from the
fact that, for disjoint~$X_1,X_2$ and assignments~$\nu_1$ of~$X_1$
and~$\nu_2$ of~$X_2$ we have~$|\nu_1 \cup \nu_2| = |\nu_1| + |\nu_2|$, and~$\vars(\phi_{g_1})$ and~$S_1$ are disjoint.} 
to \[\sum_{i\, =\, \max(0,\, \ell - |S_1|)}^{\min(\ell,\, |\vars(g_1)|) }
\alpha^i_{g_1} \times \binom{|S_1|}{\ell-i}.\]

    \item[Decomposable $\land$-gate.] If $g$ is a decomposable~$\land$-gate 
with no input then~$\phi_g$ is the Boolean function on variables\linebreak $\vars(g) =
\emptyset$ that is always true, hence~$\alpha_g^0=1$.
Otherwise, let~$g_1$ and~$g_2$ be the two input gates of~$g$. Since~$g$ is decomposable we have~$\vars(g)=\vars(g_1) \cup \vars(g_2)$ with the union being disjoint.
But then we have:
\[\sat(\phi_g) = \sat(\phi_{g_1}) \otimes \sat(\phi_{g_2})\]
We now intersect with the set of assignments of~$\vars(g)$ of size~$l$ to obtain
\begin{align*}
\alpha_g^\ell = \ssat_\ell(\phi_g) = \sum_{i\, =\, \max(0,\, \ell - |\vars(g_2)|)}^{\min(\ell,\, |\vars(g_1)|)} \alpha_{g_1}^i \times \alpha_{g_2}^{\ell-i}\
\end{align*}
This concludes the proof of the lemma, as well as the proof of Proposition~\ref{prp:shap-via-kc}.\qedhere
\end{description}
\end{proof}

\paragraph{Algorithm.\,}
Algorithm~\ref{algo:main} depicts the solution underlying Proposition~\ref{prp:shap-via-kc}.  The subroutine
ComputeAll\#SAT$_k$ takes as input a d-D
circuit~$C$ and outputs 
all the values 
$\ssat_0(C),\ldots,\ssat_{|\vars(C)|}(C)$. This function computes values~$\alpha_g^\ell$ by
bottom-up induction on~$C$ just as in the proof of Lemma~\ref{lem:ssat-kc}, by
using the appropriate equations depending on the type of each gate. 
Then, Lines~\ref{line:bb}--\ref{line:return2} in the algorithm simply
follow the part of the proof that starts at the beginning of this section until
Lemma~\ref{lem:ssat-kc}.  For instance, the returned value on
Line~\ref{line:return2} corresponds to Equation~\eqref{eq:shap-ssat}.
A quick inspection of Algorithm~\ref{algo:main} reveals that, if one ignores
the complexity of performing arithmetic operations (i.e., considering that
additions and multiplications take constant time), the running
time is~$O(|C|\cdot |\dn|^2)$. If one wishes to compute the Shapley value of
every endogenous fact (as will be done in the experiments), then the overall complexity is~$O(|C|\cdot |\dn|^3)$. \revision{Last, we point out that, in the case of non-Boolean queries, this cost is incurred for each potential output tuple that one wants to analyze.}

\IncMargin{1em}
\begin{algorithm}
\caption{Shapley values from deterministic and decomposable Boolean circuits}
\label{algo:main}
\SetKwInOut{Input}{Input}\SetKwInOut{Output}{Output}
\Input{\ Deterministic and decomposable Boolean circuit~$C$ with output gate~$g_\out$ representing~$\elin(q,\dx,\dn)$ and an endogenous fact~$f\in \dn$.}
\Output{\ The value $\shap(q,\dx,\dn,f)$.}
\SetKwFunction{FSub}{ComputeAll\#SAT$_k$}
\BlankLine
\hrulealg
\BlankLine
\BlankLine
Complete~$C$ so that~$\vars(g_\out) = \dn$\label{line:bb}\;
Compute $C_1 = C[f \rightarrow 1]$ and $C_2 = C[f \rightarrow 0]$;\tcp*[f]{Partial evaluations of~$C$ by setting~$f$ to~$1$ and to~$0$}
\BlankLine
$\Gamma = \text{ComputeAll\#SAT$_k$}(C_1)$;\tcp*[f]{As an array}

$\Delta= \text{ComputeAll\#SAT$_k$}(C_2)$;\tcp*[f]{As an array}

\BlankLine
\Return ${\displaystyle \sum_{k=0}^{|\dn|-1} \frac{k! \, (|\dn| - k - 1)!}{|\dn|!} \cdot (\Gamma[k] - \Delta[k]) }$\label{line:return2}\;
\BlankLine
\BlankLine
\SetKwProg{Fn}{Def}{:}{}
  \Fn{\FSub{C}}{
Preprocess $C$ so that each $\lor$-gate and $\land$-gate has fan-in exactly 0 or 2\label{line:preprocess}\;
\BlankLine
Compute the set~$\vars(g)$ for every gate~$g$ in $C$\;
\BlankLine
Compute values~$\alpha_g^{\ell}$ for every gate $g$ in $C$ and~$\ell\in \{0,\ldots,|\vars(g)|\}$ by bottom-up induction on~$C$ using the inductive relations from the proof of Lemma~\ref{lem:ssat-kc}\;
\BlankLine
\Return $[\alpha_{g_\out}^0,\ldots,\alpha_{g_\out}^{|\vars(C)|}]$\;
  }
\end{algorithm}
\DecMargin{1em}

\subsection{Implementation Architecture}
\label{subsec:strategy}

In this section, we present our architecture for implementing the knowledge compilation
approach over realistic datasets. The relevant parts, for now, are the middle and top
part of Figure~\ref{fig:strategy}, which we next explain.
Given a database~$D=\dx \cup \dn$, a query~$q(\bar{x})$, a tuple~$\bar{t}$ of
the same arity as~$\bar{x}$, and an endogenous fact~$f\in \dn$, we want to
compute~$\shap(q[\bar{x}/\bar{t}],\dx,\dn,f)$.  We use two existing tools to
help us with this task: ProvSQL~\cite{senellart2018provsql} and the knowledge
compiler c2d~\cite{darwiche2004new}. ProvSQL is a tool integrated into PostgreSQL
that can perform provenance (lineage) computation in various semirings.  For
our purposes, a knowledge compiler is a tool that takes as input a Boolean
function in CNF and outputs an equivalent Boolean function into another
formalism. The target formalism that we will use is the so-called ``d-DNNF". A
d-DNNF is simply a deterministic and decomposable Boolean circuit such that negation
gates are only applied to variables (NNF stands for \emph{negation normal
form}).\footnote{This additional NNF restriction is not important here, but,
as far as we know, no knowledge compiler has the more general
``deterministic and decomposable circuits'' (without NNF) as a target. \revision{It is currently unknown whether d-DNNFs and deterministic and decomposable
circuits are exponentially separated or not~\cite[Table 7]{darwiche2002knowledge}.}}

In our case, we use ProvSQL as follows: we feed it the database~$D$,
query~$q(\bar{x})$ and tuple~$\bar{t}$, and ProvSQL
computes~$\lin(q[\bar{x}/\bar{t}],D)$ as a Boolean circuit, called~$C$ 
in Figure~\ref{fig:strategy}.  
We note here that for SPJU queries, $C$ can be computed in polynomial-time data complexity.
We then set to~$1$ all the exogenous facts to obtain a Boolean
circuit~$C'$ for~$\elin(q[\bar{x}/\bar{t}],\dx,\dn)$.  Then, ideally, we would
like to use the knowledge compiler to transform~$C'$ into an equivalent d-DNNF,
in order to be able to apply Algorithm~\ref{algo:main}. Unfortunately, every
knowledge compiler that we are aware of takes as input Boolean formulas in
conjunctive normal form (CNF), and not arbitrary Boolean circuits. To
circumvent it, we use the \emph{Tseytin
transformation}~\cite{tseitin1983complexity} to transform the circuit~$C'$ into
a CNF~$\phi \defeq \textsc{Tseytin}(C')$, whose size is linear in that of~$C'$. This CNF~$\phi$ has the following
properties: (1) its variables are the variables of~$C'$ plus a set~$Z$ of additional
variables; 
(2) for every valuation~$\nu
\subseteq \vars(C')$ that satisfies~$C'$, there exists exactly one
valuation~$\nu' \subseteq Z$ such that~$\phi(\nu \cup \nu')=1$; and (3) for every
valuation~$\nu \subseteq \vars(C')$ that does not satisfy~$C'$, there is no
valuation~$\nu' \subseteq Z$ such that~$\phi(\nu \cup \nu')=1$. We then
feed~$\phi$ to the knowledge compiler, which produces a d-DNNF~$C''$ 
equivalent to~$\phi$ (the variables
of~$C''$ are again $\vars(C') \cup Z$).  We note here that there is no
theoretical guarantee that this step is efficient; indeed,
the task of transforming a CNF into an equivalent d-D
circuit is~$\fptoshp$-hard in general; see Section \ref{sec:experiments} for an experimental analysis of its tractability in practice.  Next, we need to
eliminate the additional variables~$Z$ in order to be able to apply
Algorithm~\ref{algo:main}. To this end, we use the following Lemma.

\begin{lemma}
\label{lem:remove-extra}
Given as input a d-DNNF~$C''$ that is equivalent to $\textsc{Tseytin}(C')$ for
a Boolean circuit~$C'$, we can compute in time~$O(|C''|)$ a d-DNNF~$C'''$
that is equivalent to~$C'$ (in particular, the variables of~$C'''$ are the same
as the variables of~$C'$; in our case, they consist only of endogenous facts).
\end{lemma}
\begin{proof}[Proof sketch]
Let~$Z$ be the additional variables coming from the Tseytin transformation. First, we
remove all the gates of~$C''$ that are not satisfiable,
and then we remove all the gates that are not connected to the output gate.  Now, let~$C'''$ be the circuit
that is obtained from this intermediate circuit by replacing every literal~$z$ or~$\lnot z$ for~$z
\in Z$ by a constant~$1$-gate. We return~$C'''$.  The proof that this algorithm
is correct uses the properties (1--3) of the Tseytin transformation, and is omitted due to lack of space.
\end{proof}

Using this lemma, we obtain a d-DNNF~$C'''$ for the endogenous lineage
$\elin(q[\bar{x}/\bar{t}],\dx,\dn)$, to which we can finally apply
Algorithm~\ref{algo:main} to obtain the value
$\shap(q[\bar{x}/\bar{t}],\dx,\dn,f)$.  

\begin{figure*}
\scalebox{0.85}{
  \input{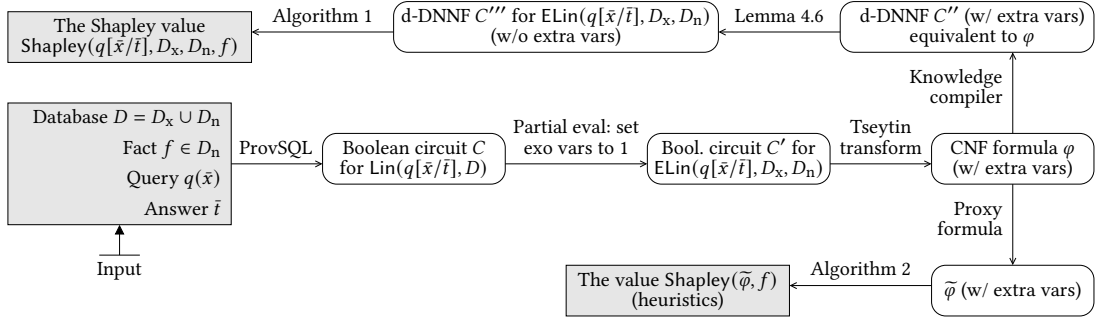}
  }
  \caption{Our implementation architecture.}
\label{fig:strategy}
\vskip-1em
\end{figure*}

\section{Inexact Computation}
\label{sec:approx}

As we will show in the experimental section, the exact computation algorithm that we have proposed performs well in most cases but is too costly in others. In the latter cases, we may wish to resort to methods that do not necessarily compute exact Shapley values, if their results still typically suffice to determine the {\em order} of facts according to their Shapley contribution.
In this section we propose \emph{\cprox{}}, a heuristic solution that is very efficient, and we will experimentally show that  the ranking of facts based on \cprox{} tends to match the ranking based on the exact Shapley values.

At a high level, \cprox{} is based on the observation that having a high Shapley score is correlated  (albeit in a complex manner) with (1) appearing many times in the provenance and (2) having few ``alternatives,'' that is, facts that could compensate for the absence of the given fact. The first factor (number of occurrences) may be directly read from the CNF obtained by applying the Tseytin transformation to the provenance circuit ($C'$ in Figure~\ref{fig:strategy}). It is also easy to read from the CNF partial information about the second factor (number of alternatives), namely the number of alternatives in each clause (ignoring intricate dependencies between clauses).  
Next, we present the details of \cprox.

\IncMargin{1em}
\begin{algorithm}[h]
    \SetKwFunction{CNF Proxy}{CNF Proxy}
    \SetKwInOut{Input}{Input}\SetKwInOut{Output}{Output}
    \LinesNumbered
    \Input{\ CNF $\varphi$ and a set of endogenous facts $\dn$.}
    \Output{\ The value $\shap(\widetilde{\varphi}, x)$ for each $x \in \dn$.} \BlankLine
    \hrulealg
    \BlankLine
    $n \gets \lvert \varphi.clauses() \rvert$\;\label{line:cprox_count_clauses}
    
    $v \gets 0^{|\dn|}\text{;}$ \tcp*[f]{As an array}\label{line:cprox_init}
    
    \For{$\psi \in \varphi.clauses()$}{
    \label{line:cprox_loop}
        $\mathcal{L} \gets \psi.literals()$\;\label{line:cprox_literals}
        
        $m \gets \lvert \mathcal{L} \rvert$\;\label{line:cprox_literals_count}
        
        $pos \gets \{\ell \in \mathcal{L} \mid \ell \text{ is positive} \}$\;\label{line:cprox_pos}
        
        $neg \gets \{\ell \in \mathcal{L} \mid \ell \text{ is negative} \}$\;\label{line:cprox_neg}
        
        \For{$\ell \in pos \cap \dn$}{
            $v[\ell.var()] \gets v[\ell.var()] + \frac{1}{nm \cdot \binom{m-1}{\lvert neg \rvert}}$\;
            \label{line:cprox_add}
        }
        \For{$\ell \in neg  \cap \dn$}{
            $v[\ell.var()] \gets v[\ell.var()] - \frac{1}{nm \cdot \binom{m-1}{\lvert pos \rvert}}$\;
            \label{line:cprox_remove}
        }
    }
    \Return $v$\label{line:cprox_retufn}
\caption{CNF Proxy}
\label{algo:cnf_proxy}
\end{algorithm} 
\DecMargin{1em}

We will start with an auxiliary definition, denoting the Shapley value of a general function $h : 2^X \to \R$ and a variable $x \in X$ as

\begin{align*}
\label{eq:defshapbool}
\shap(h, x) \ \defeq \  & \sum_{S \subseteq X \setminus \{x\}} \frac{|S|!(|X|-|S|-1)!}{|X|!} \big(h(S \cup \{x\}) - h(S)\big).
\end{align*}
Naturally, if $h = \elin(q,\dn,\dx)$, i.e., the endogenous lineage, and~$x$ is a fact in $\dn$, then $\shap(h, x) = \shap(q, \dn, \dx, x)$.

Now, note that for a CNF formula $\varphi=\bigwedge_{i=1}^n \psi_i$ (where each $\psi_i$ is a disjunction of literals) and an assignment $\nu$ it holds that $\varphi(\nu) = \prod_{i=1}^n \psi_i(\nu)$.
Instead of calculating the Shapley values of a CNF formula~$\phi$ (which may be a hard problem), \cprox{} computes Shapley values with respect to a \emph{proxy function}, denoted~$\widetilde{\varphi}$. The proxy function of $\varphi$ is defined as the sum (instead of the product) of the clauses of~$\varphi$, i.e., $\widetilde{\varphi}(\nu) \defeq \sum_{i=1}^n \frac{1}{n} \psi_i(\nu)$. Intuitively, a fact that appears in many clauses of the CNF $\varphi$ will occur in many summands of $\widetilde{\varphi}$, and when we compute Shapley values with respect to $\widetilde{\varphi}$, the number of alternatives in each clause will be reflected in decreased value of the respective summands.

\begin{example}
\label{ex:proxfirst}
Consider the CNF formula $\varphi = (x_1 \vee x_2) \wedge (x_1 \vee x_3 \vee x_4)$. The Shapley values of $x_1, x_2, x_3, x_4$ are $\frac{7}{12}, \frac{3}{12}, \frac{1}{12}, \frac{1}{12}$ respectively.
Note that $x_1$ has the highest influence, which intuitively may be attributed to its appearance in two clauses whereas each other variable appears only in a single clause. The variable $x_2$ has more influence than $x_3$ and $x_4$, intuitively since it has less alternatives. These comparative features are preserved in $\widetilde{\varphi} = (x_1 \vee x_2) + (x_1 \vee x_3 \vee x_4)$, and indeed the  Shapley values of $x_1, x_2, x_3, x_4$ with respect to $\widetilde{\varphi}$ are $\frac{5}{6}, \frac{1}{2}, \frac{1}{3}, \frac{1}{3}$ respectively. Observe that although the values assigned to the variables are very different from their actual Shapley values, their order remains intact in this case.
\end{example}

Due to the linearity of Shapley values, their computation with respect to $\widetilde{\varphi}$ is much more efficient, as implied by the following lemma
(whose proof we omit for space reasons):

\begin{lemma}\label{lem:proxy}
Let $h = \sum_{i=1}^n \frac{1}{n} \psi_i$, where each $\psi_i$ is
a Boolean function representing a disjunction of literals. Without loss of
generality let us assume that for each $\psi_i$ there is no variable that
appears in more than one literal of~$\psi_i$.  Denote by $a_i$ and $b_i$ the
number of positive and negative literals in $\psi_i$ respectively. Then, for
every variable~$x$, it holds that $\shap(h,x) = \frac{1}{n} \sum_{i=1}^n
\Phi(\psi_i, x)$, where 
\[
\Phi(\psi_i, x) =  
\begin{cases}
  \frac{1}{(a_i + b_i) \cdot \binom{a_i + b_i - 1}{b_i}} & \text{if $x$ appears in $\psi_i$ in positive form;} \\
  \frac{-1}{(a_i + b_i) \cdot \binom{a_i + b_i - 1}{a_i}} & \text{if $x$ appears in $\psi_i$ in negative form;} \\  
  0 & \text{otherwise.}
\end{cases}
\]
\end{lemma}

Algorithm~\ref{algo:cnf_proxy} then describes the operation of \cprox{}, computing Shapley values of~$\widetilde{\varphi}$ according to Lemma~\ref{lem:proxy}. The input to the algorithm is a CNF formula $\varphi$ and the set of endogenous facts $\dn$. 
In Line~\ref{line:cprox_count_clauses}, \cprox{} counts the number~$n$ of clauses of~$\phi$, and in Line~\ref{line:cprox_init} it initializes the contribution of every variable of~$\phi$ to zero. Then, it iterates over the clauses (Line~\ref{line:cprox_loop}). For each clause, it counts the number~$m$ of literals, and identifies the set of positive and negative literals, $pos$ and $neg$ respectively (Lines~\ref{line:cprox_literals}--\ref{line:cprox_neg}). In Lines~\ref{line:cprox_add} and ~\ref{line:cprox_remove} we add (resp., subtract) quantities for each variable in a positive (resp., negative) literal according to Lemma ~\ref{lem:proxy} (note that $m$ corresponds to $a_i+b_i$ as denoted in the lemma). Finally, in Line~\ref{line:cprox_retufn} the algorithm returns the contribution value of each fact $x\in \dn$ based on $\shap(\widetilde{\varphi}, x)$.
Observe that \cprox{} runs in linear time in the size of~$\varphi$ (which itself, being the 
Tseytin transformation of the provenance circuit $C'$ from Figure~\ref{fig:strategy}, is linear in~$C'$).

\begin{example}
\label{ex:cprox}
Recall the queries $q_1, q_2, q$ and their lineages depicted in Figure~\ref{fig:airports_example}. 
Using endogenous lineages, we have:
\begin{itemize}
    \item $\elin(q_1, \dn, \dx) = a_1$
    \item 
    $
    \!
    \begin{aligned}[t]
    \elin(q_2, \dn, \dx) = &\left( a_2 \wedge a_4 \right) \vee \left( a_2 \wedge a_5 \right) \vee \left( a_3 \wedge a_4 \right) \vee \\
    & \left( a_3 \wedge a_5 \right) \vee \left( a_6 \wedge a_7 \right)
    \end{aligned}
    $
    \item
    $
    \!
    \begin{aligned}[t]
    \elin(q, \dn, \dx) = &a_1 \vee \left( a_2 \wedge a_4 \right) \vee \left( a_2 \wedge a_5 \right) \vee \left( a_3 \wedge a_4 \right) \vee \\
    & \left( a_3 \wedge a_5 \right) \vee \left( a_6 \wedge a_7 \right)
    \end{aligned}
    $
\end{itemize}
The lineage of $q_1$ is a CNF with a single variable, thus the contribution of $a_1$ to $q_1$ as computed by Algorithm~\ref{algo:cnf_proxy} is~$1$, which is indeed equal to $\shap(q_1, \dn, \dx, a_1)$.
Applying the Tseytin transformation to the lineage of~$q_2$ introduces $6$ new variables ($\{z_i\}_{i=1}^6$) and results in the following equisatisfiable CNF:
{\footnotesize
\begin{align*}
& \left( z_1 \right) \wedge \left( z_1 \vee \conj{z_2} \right) \wedge \left( z_1 \vee \conj{z_3} \right) \wedge \left( z_1 \vee \conj{z_4} \right) \wedge \left( z_1 \vee \conj{z_5} \right) \wedge \left( z_1 \vee \conj{z_6} \right) \wedge \\
&\left( \conj{z_1} \vee z_2 \vee z_3 \vee z_4 \vee z_5 \vee z_6 \right) \wedge \\
& \left( \conj{z_2} \vee a_2 \right) \wedge \left( \conj{z_2} \vee a_4 \right) \wedge \left(z_2 \vee \conj{a_2} \vee \conj{a_4}\right) \wedge
\left( \conj{z_3} \vee a_2 \right) \wedge \left( \conj{z_3} \vee a_5 \right) \wedge \left(z_3 \vee \conj{a_2} \vee \conj{a_5} \right) \wedge \\
& \left( \conj{z_4} \vee a_3 \right) \wedge \left( \conj{z_4} \vee a_4 \right) \wedge \left(z_4 \vee \conj{a_3} \vee \conj{a_4} \right) \wedge \left( \conj{z_5} \vee a_3 \right) \wedge \left( \conj{x_5} \vee a_5 \right) \wedge \left(z_5 \vee \conj{a_3} \vee \conj{a_5}\right) \wedge \\ 
& \left(\conj{z_6} \vee a_6 \right) \wedge \left( \conj{z_6} \vee a_7 \right) \wedge \left(z_6 \vee \conj{a_6} \vee \conj{a_7} \right)
\end{align*}
}
Algorithm~\ref{algo:cnf_proxy} iterates over the above clauses, and computes the contribution of the endogenous facts $\dn$ over the proxy function. Note that the facts $\dn$ appear in clauses of two forms. 
The first form is $\left(\conj{z_j} \vee a_i\right)$; appearance in this type of clause adds $\frac{1}{22 \cdot 2 \cdot \binom{1}{1}} = \frac{1}{44}$ to the contribution of~$a_i$. The second form is $\left( z_j \vee \conj{a_i} \vee \conj{a_h}\right)$, which adds $\frac{-1}{22 \cdot 3 \cdot \binom{2}{1}} = \frac{-1}{132}$.
Note  that each of $a_2, a_3, a_4, a_5$ has two appearances in clauses of the first form, and one appearance in clauses of the second form. Thus, according to Algorithm~\ref{algo:cnf_proxy} the contribution of $a_2, a_3, a_4, a_5$ is $\frac{5}{132} \approx 0.038$.
In contrast, $a_6$ and $a_7$ each have a single appearance in a clause of the first form and a single appearance in a clause of the second form, thus their contribution is $\frac{1}{66} \approx 0.015$.
We note that the values calculated by Algorithm~\ref{algo:cnf_proxy} are very different from the actual Shapley values, as $\shap(q_2, \dn, \dx, a_i)=\frac{11}{60} \approx 0.183$ for $a_i \in \{a_2, a_3, a_4, a_5\}$ and $\shap(q_2, \dn, \dx, a_i)=\frac{2}{15} \approx 0.133$ for $a_i \in \{a_6, a_7\}$. However, the facts $a_2, a_3, a_4, a_5$, are correctly determined to be more influential than $a_6$ and $a_7$.

\end{example}

Our experimental evaluation indicates that in most cases, the ordering of facts according to the values assigned to them by Algorithm~\ref{algo:cnf_proxy} agrees with the order obtained by using the actual Shapley values. There is however no theoretical guarantee that this will always be the case, as shown by the following example.   

\begin{example}
\label{ex:proxfar}
Applying the Tseytin transformation over the lineage of~$q$ will result in a CNF similar to that obtained for $q_2$, with a new variable ($z_7$) and the new clauses 
$\left( z_1 \vee \conj{z_7} \right) \wedge \left( z_7 \vee \conj{a_1} \right) \wedge \left( \conj{z_7} \vee a_1 \right)$.
In addition, the disjunct $\vee z_7$ will be added to the clause\linebreak $\left( \conj{z_1} \vee z_2 \vee z_3 \vee z_4 \vee z_5 \vee z_6 \right)$. Similarly to the case of $q_2$, the contributions of $a_2, a_3, a_4$, and $a_5$ are correctly determined to be larger than those of $a_6$ and $a_7$. As for $a_1$, its contribution according to Algorithm~\ref{algo:cnf_proxy} is 0 while in fact it is the most influential fact.

\end{example}

\section{Experiments}\label{sec:experiments}

Our system is implemented in Python 3.6 and using the PostgreSQL 11.10 database engine, \revision{and the experiments were performed on a 
Linux Debian 14.04 machine with 1TB of RAM and an Intel(R) Xeon(R) Gold 6252 CPU @ 2.10GHz processor}.
ProvSQL~\cite{senellart2018provsql} was used to capture the provenance. For knowledge compilation we have used the c2d compiler \cite{darwiche2001tractability,darwiche2004new}. The source code of our implementation is available in \cite{frost2021github}.

Since no standard benchmark for our problem exists, we have created such a benchmark of 40 queries over the TPC-H (1.4GB) and IMDB (1.2GB) databases.\revision{The TPC-H queries are based on the ones in \cite{tpch}, where we have only removed nested queries (which ProvSQL does not handle) and aggregation operations (for which provenance is not Boolean). The queries for the IMDB database are based on the join queries in \cite{leis2015good}, where for each query we have added a (last) projection operation over one of the join attributes to make provenance more complex and thus more challenging for our algorithms. The resulting queries are quite complex: in particular, only 4 out of the 40 are hierarchical.  See \cite{frost2021github} for details of the obtained queries.}

\subsection{Exact computation}
\label{sec:exp_exact}

We have evaluated our solution for exact Shapley computation, presented in Section~\ref{subsec:strategy}, on each of the 40 queries.
In total, we have obtained 95,803 output tuples
along with their provenance expressions (computed with ProvSQL). We have then transformed each
provenance expression into a d-DNNF structure using the c2d
\cite{darwiche2001tractability,darwiche2004new} knowledge compiler. \revision{In this experiment}, for both
the knowledge compilation (KC) and Shapley evaluation steps we have set a timeout
of \revision{one hour (in fact, as we show below, a much shorter timeout of 2.5 seconds typically suffices)}. In case the compilation completed successfully within this
timeframe, we have computed the Shapley values using
Algorithm~\ref{algo:main}.  Table \ref{tab:exp_exact} presents the execution times of our solution for 16 representative queries; next we will overview different
aspects of the results.

\begin{table*}[!htb]
    \centering
    \footnotesize
    \caption{\revision{Statistics on the exact computation of Shapley values for 16 representative queries}}
    \label{tab:exp_exact}
    \begin{tabular}{| c | c | c c c c c| c c c c c | c c c c c |}
        \hline
         \multirow{2}{*}{\textbf{Dataset}} & \multirow{2}{*}{\textbf{Query}} &
         \multirow{2}{*}{\makecell{\textbf{\#Joined} \\ \textbf{tables}}} & 
         \multirow{2}{*}{\makecell{\textbf{\#Filter} \\ \textbf{conditions}}} &         
         \multirow{2}{*}{\makecell{\textbf{Execution} \\ \textbf{time [sec]}}} &
         \multirow{2}{*}{\makecell{\textbf{\#Output} \\ \textbf{tuples}}} & 
         \multirow{2}{*}{\makecell{\textbf{Success} \\ \textbf{rate}}} &
         \multicolumn{5}{c |}{\textbf{KC execution times [sec]}} & \multicolumn{5}{c |}{\textbf{Alg.~\ref{algo:main} execution times [sec]}}\\
         & & & & & & & \textbf{Mean} &  \textbf{p25} & \textbf{p50} & \textbf{p75} & \textbf{p99} & \textbf{Mean} & \textbf{p25} & \textbf{p50} & \textbf{p75} & \textbf{p99} \\
         \hline
         \parbox[|c|]{2mm}{\multirow{8}{*}{\rotatebox[origin=c]{90}{\texttt{TPC-H}}}} 
        
        \revision{} & \revision{3} & \revision{3} & \revision{5} & \revision{20980.71} &
        \revision{100} & \revision{100\%} & \revision{0.06} & 
        \revision{0.04} & \revision{0.07} & \revision{0.08} & \revision{0.13} & \revision{0.00} & 
        \revision{0.00} & \revision{0.00} & \revision{0.00} & \revision{0.01}\\
        
        \revision{} & \revision{5} & \revision{6} & \revision{9} & \revision{48.67} & \revision{5} & \revision{0\%} & \revision{-} & 
        \revision{-} & \revision{-} & \revision{-} & \revision{-} & 
        \revision{-} & \revision{-} & \revision{-} & \revision{-} & \revision{-}\\
        
         \revision{} & \revision{7} & \revision{6} & \revision{8} & \revision{30.57} & \revision{4} & \revision{0\%} & \revision{-} & 
        \revision{-} & \revision{-} & \revision{-} & \revision{-} & \revision{-} & 
        \revision{-} & \revision{-} & \revision{-} & \revision{-}\\

         \revision{} & \revision{10} & \revision{4} & \revision{6} & \revision{4.72} & \revision{10} & \revision{100\%} & \revision{0.14} & 
        \revision{0.13} & \revision{0.13} & \revision{0.14} & \revision{0.19} & \revision{0.01} &
        \revision{0.01} & \revision{0.01} & \revision{0.01} & \revision{0.01}\\

         \revision{} & \revision{11} & \revision{6} & \revision{7} & \revision{0.13} & \revision{10} & \revision{100\%} & \revision{0.09} & 
        \revision{0.08} & \revision{0.08} & \revision{0.08} & \revision{0.13} & \revision{0.01} & 
        \revision{0.00} & \revision{0.00} & \revision{0.01} & \revision{0.03}\\

         \revision{} & \revision{16} & \revision{3} & \revision{5} & \revision{1.25} & \revision{10} & \revision{100\%} & \revision{0.26} & 
        \revision{0.13} & \revision{0.18} & \revision{0.33} & \revision{0.57} & \revision{0.18} & 
        \revision{0.01} & \revision{0.03} & \revision{0.29} & \revision{0.88}\\

         \revision{} & \revision{18} & \revision{4} & \revision{3} &  \revision{37.31} & \revision{10} & \revision{100\%} & \revision{0.13} & 
        \revision{0.08} & \revision{0.08} & \revision{0.18} & \revision{0.23} & \revision{0.01} & 
        \revision{0.00} & \revision{0.01} & \revision{0.01} & \revision{0.03}\\
         
         \revision{} & \revision{19} & \revision{2} & \revision{21} & \revision{2.04} & \revision{1} & \revision{100\%} & \revision{1.20} & 
        \revision{1.20} & \revision{1.20} & \revision{1.20} & \revision{1.20} & \revision{156.06} & 
        \revision{156.06} & \revision{156.06} & \revision{156.06} & \revision{156.06}\\

         \hline

         \parbox[c]{.5mm}{\multirow{8}{*}{\rotatebox[origin=c]{90}{\texttt{IMDB}}}}

        \revision{} & \revision{1a} & \revision{5} & \revision{10} & \revision{0.25} & \revision{35} & \revision{100\%} & \revision{0.17} & 
        \revision{0.08} & \revision{0.08} & \revision{0.13} & \revision{2.10} & \revision{5.92} &  
        \revision{0.01} & \revision{0.01} & \revision{0.01} & \revision{206.72}\\  
         
        \revision{} & \revision{6b} & \revision{5} & \revision{8} & \revision{2.61} & \revision{1} & \revision{100\%} & \revision{0.44} & 
        \revision{0.44} & \revision{0.44} & \revision{0.44} & \revision{0.44} & \revision{0.08} & 
        \revision{0.08} & \revision{0.08} & \revision{0.08} & \revision{0.08}\\

         \revision{} & \revision{7c} & \revision{8} & \revision{21} & \revision{77.33} & \revision{2415} & \revision{99\%} & \revision{0.82} & 
        \revision{0.18} & \revision{0.28} & \revision{0.63} & \revision{9.44} & \revision{24.28} & 
        \revision{0.02} & \revision{0.05} & \revision{0.39} & \revision{787.12}\\
         
         \revision{} & \revision{8d} & \revision{7} & \revision{10} &  \revision{145.10} & \revision{44517} & \revision{99.9\%} & \revision{0.26} & 
        \revision{0.13} & \revision{0.18} & \revision{0.29} & \revision{1.09} & \revision{0.36} & 
        \revision{0.01} & \revision{0.02} & \revision{0.05} & \revision{2.28}\\
                 
        \revision{} & \revision{11a} & \revision{8} & \revision{18} & \revision{3.20} & \revision{10} & \revision{100\%} & \revision{0.75} & 
        \revision{0.18} & \revision{0.48} & \revision{1.14} & \revision{2.15} & \revision{23.34} & 
        \revision{0.01} & \revision{0.13} & \revision{1.27} & \revision{151.99}\\
         
         \revision{} & \revision{11d} &  \revision{8} & \revision{16} & \revision{56.99} & \revision{210} & \revision{98.1\%} & \revision{0.99} &  
         \revision{0.19} & \revision{0.32} & \revision{0.78} & \revision{5.97} & \revision{96.19} & 
        \revision{0.02} & \revision{0.06} & \revision{0.48} & \revision{1650.72}\\

        \revision{} & \revision{13c} &  \revision{9} & \revision{19} & \revision{2.44} & \revision{14} & \revision{100\%} & \revision{0.22} & 
        \revision{0.13} & \revision{0.18} & \revision{0.28} & \revision{0.53} & \revision{0.02} & 
        \revision{0.01} & \revision{0.01} & \revision{0.01} & \revision{0.06}\\
        
        \revision{} & \revision{15d} & \revision{9} & \revision{18} & \revision{24.25} & \revision{207} & \revision{97.6\%} & \revision{1.89} & 
        \revision{0.24} & \revision{0.48} & \revision{1.28} & \revision{10.01} & \revision{70.05} &  
        \revision{0.06} & \revision{0.30} & \revision{3.52} & \revision{1821.13}\\
         
        \revision{} & \revision{16a} & \revision{8} & \revision{15} & \revision{5.56} & \revision{173} & \revision{100\%} & \revision{0.18} & 
        \revision{0.13} & \revision{0.14} & \revision{0.20} & \revision{0.53} & \revision{0.02} &  
        \revision{0.01} & \revision{0.01} & \revision{0.02} & \revision{0.12}\\
        
        \hline
    \end{tabular}
\end{table*}

\smallskip

\revision{
\textit{Success rate.} 
We report here the rate of successful executions.
The IMDB queries resulted in 95,636 output tuples; the KC step completed successfully for 95,599 out of 
them, where all 37 failures were the result of insufficient memory. For each of the IMDB output tuples that were successfully compiled into a d-DNNF, we have executed Algorithm~\ref{algo:main};
only a single execution failed in this step (due to a timeout of one hour). Overall, the exact computation of Shapley values was successful for 95,598 out of the 95,636 IMDB output tuples (i.e., 99.96\% success rate).
The TPC-H queries resulted in 167 output tuples; the KC step has completed successfully for 141 out of 
them, again all 26 failures were the result of insufficient memory. For all TPC-H outputs that compiled successfully Algorithm~\ref{algo:main} was successful, yielding an overall 84.43\% success rate.
}

\begin{figure}[t]
    \centering
    \begin{subfigure}{.48\linewidth}
    \centering
    \includegraphics[width=\linewidth]{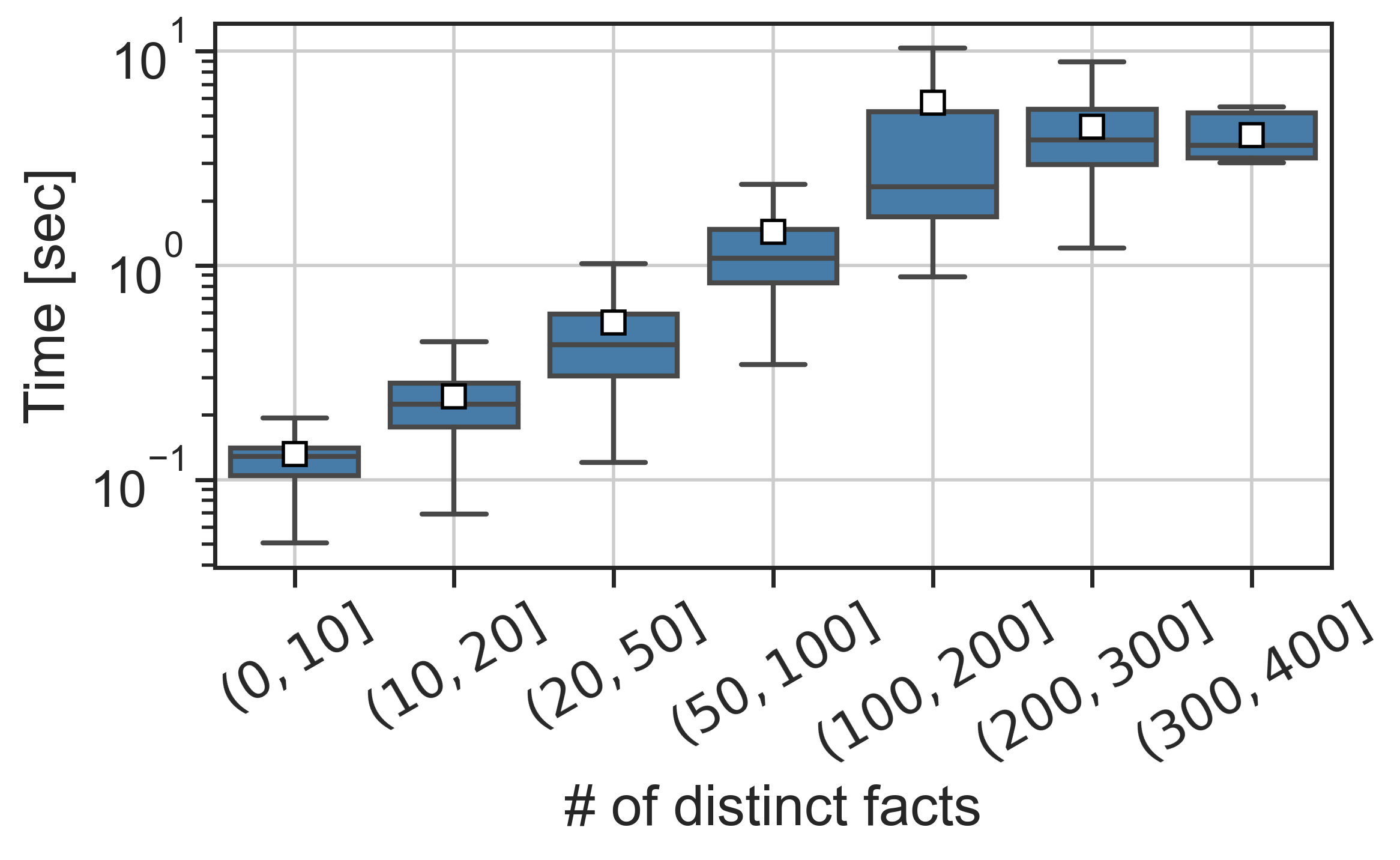}
    \caption{\revision{KC time per \#facts}}
    \label{fig:runnig_time_kc_tuples}
    \end{subfigure}%
    \begin{subfigure}{.48\linewidth}
    \centering
    \includegraphics[width=\linewidth]{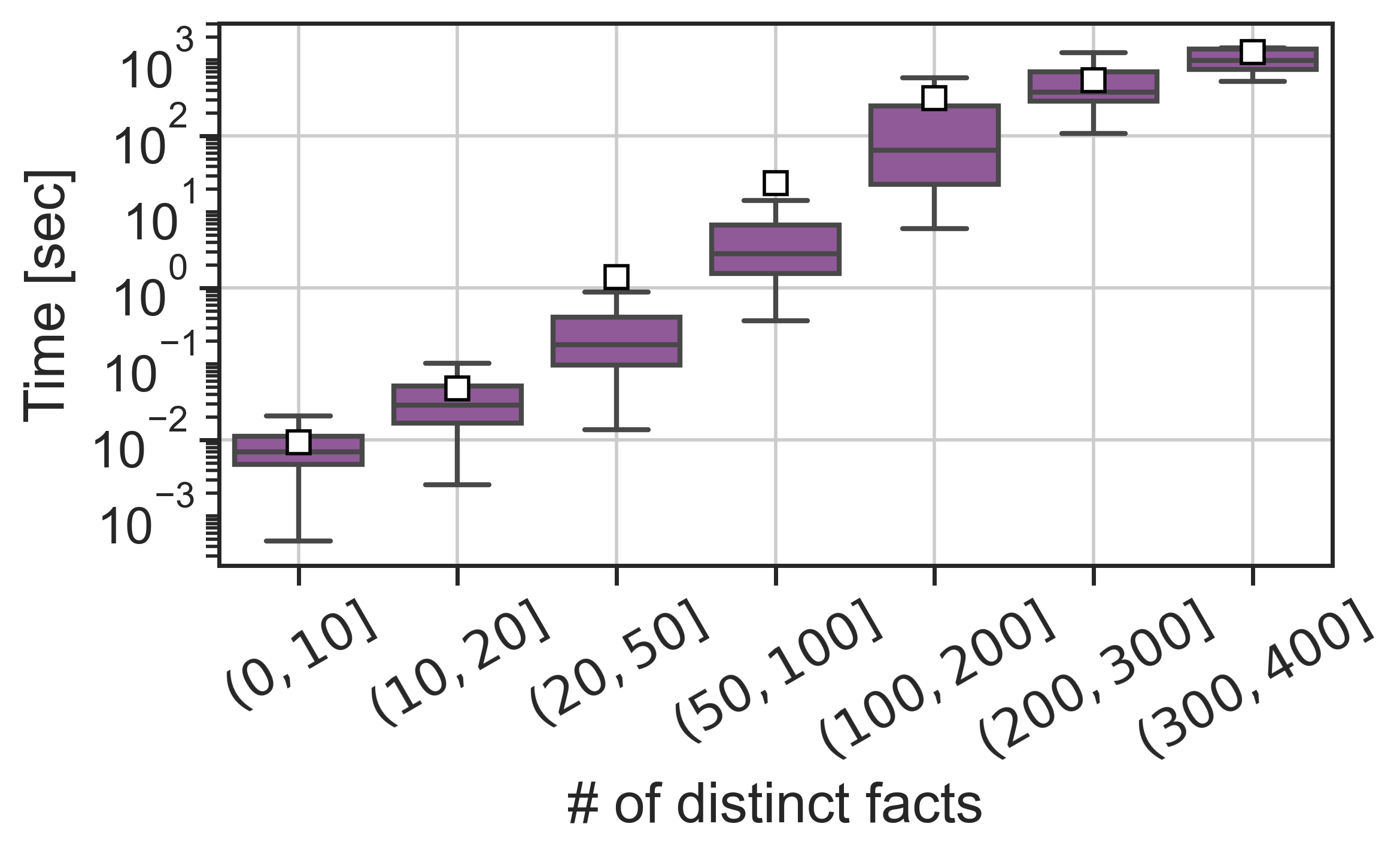}
    \caption{\revision{Alg.~\ref{algo:main} time per \#facts}}
    \label{fig:runnig_time_shapley_tuples}
    \end{subfigure}\\   \begin{subfigure}{.48\linewidth}
    \centering
    \includegraphics[width=\linewidth]{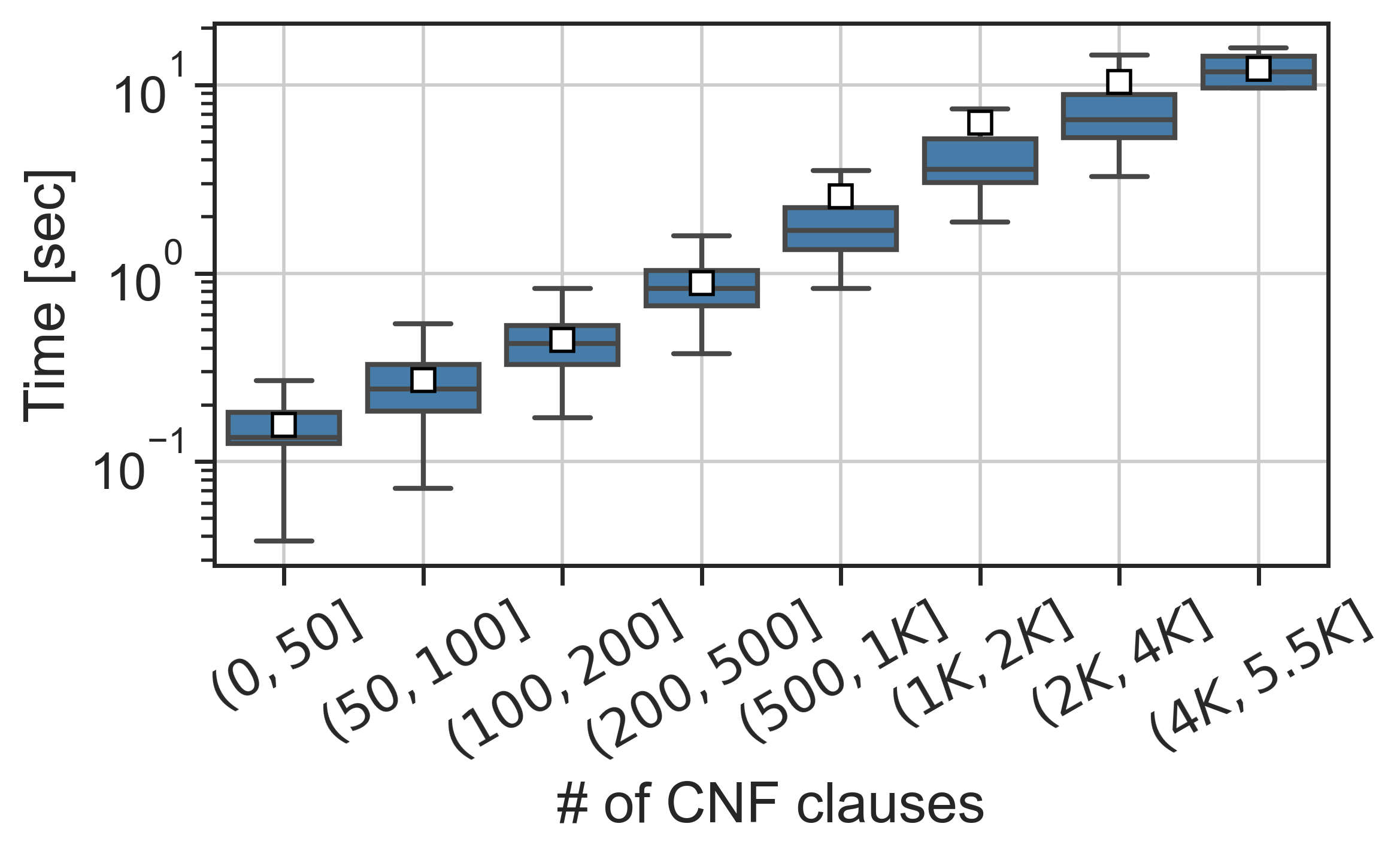}
    \caption{\revision{KC time per \#CNF clauses}}
    \label{fig:runnig_time_kc_clauses}
    \end{subfigure}%
    \begin{subfigure}{.48\linewidth}
    \centering
    \includegraphics[width=\linewidth]{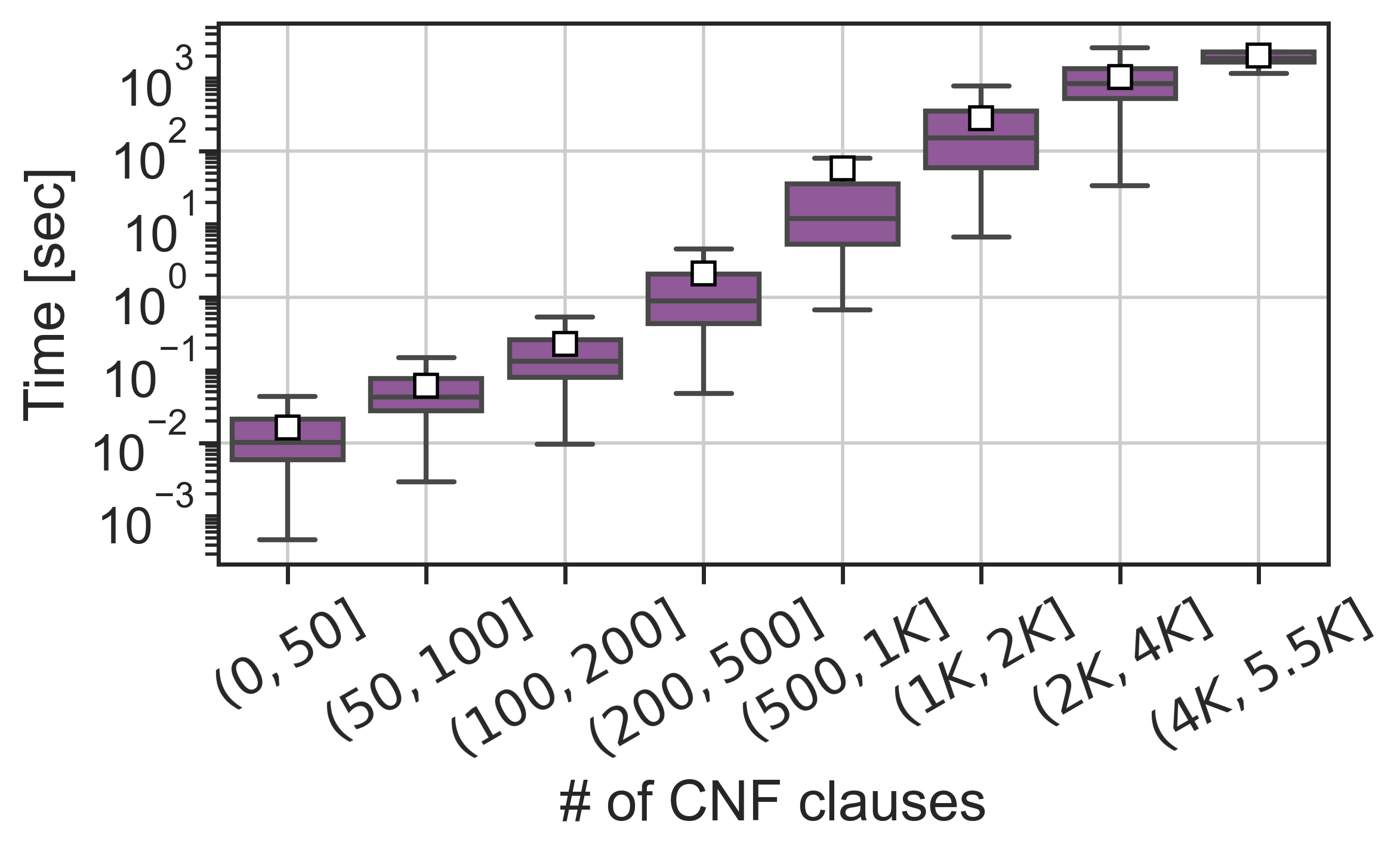}
    \caption{\revision{Alg.~\ref{algo:main} time per \#CNF clauses}}
    \label{fig:runnig_time_shapley_clauses}
    \end{subfigure}\\    
    \begin{subfigure}{.48\linewidth}
    \centering
    \includegraphics[width=\linewidth]{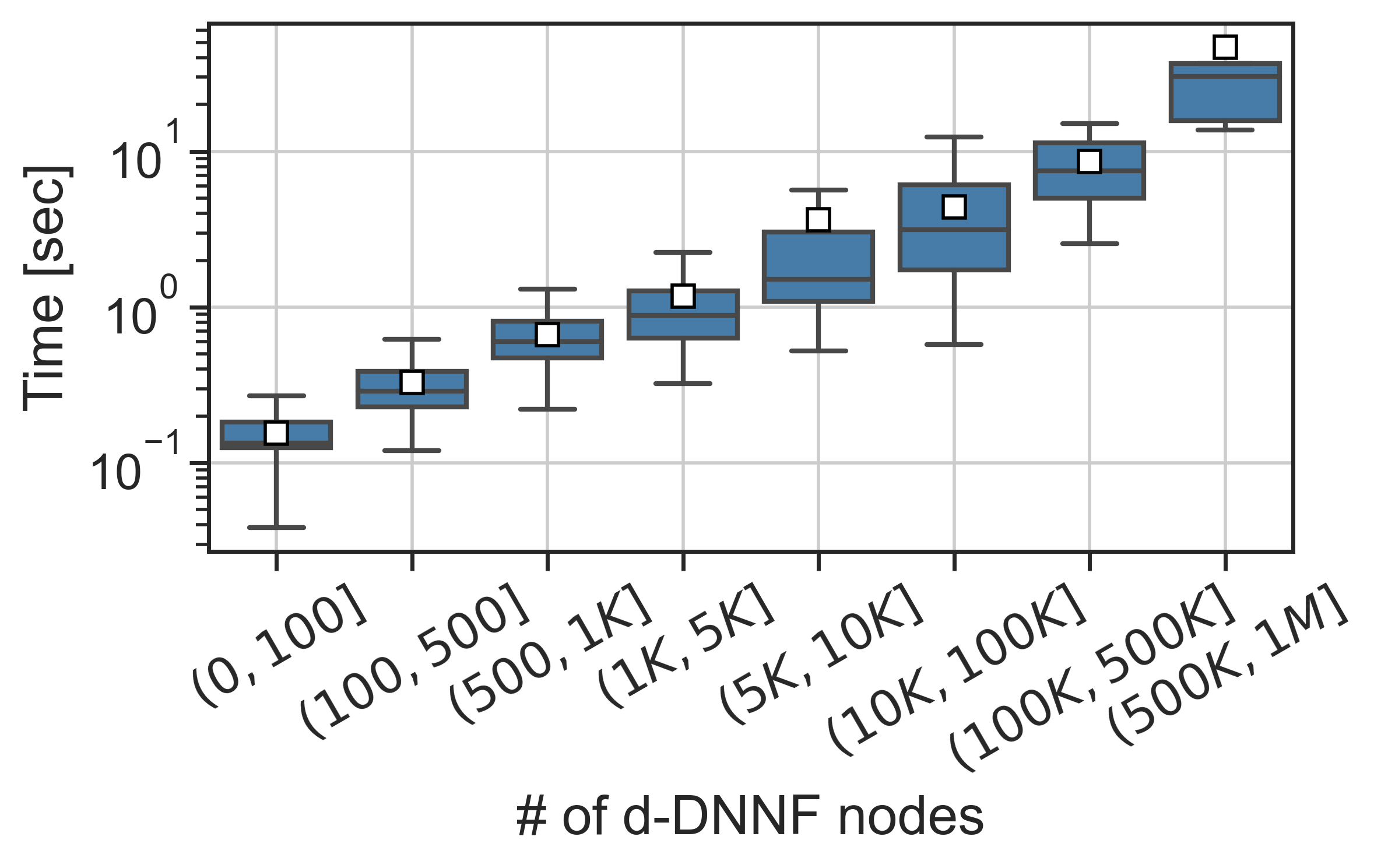}
    \caption{\revision{KC time per d-DNNF size}}
    \label{fig:runnig_time_kc_ddnnf}
    \end{subfigure}%
    \begin{subfigure}{.48\linewidth}
    \centering
    \includegraphics[width=\linewidth]{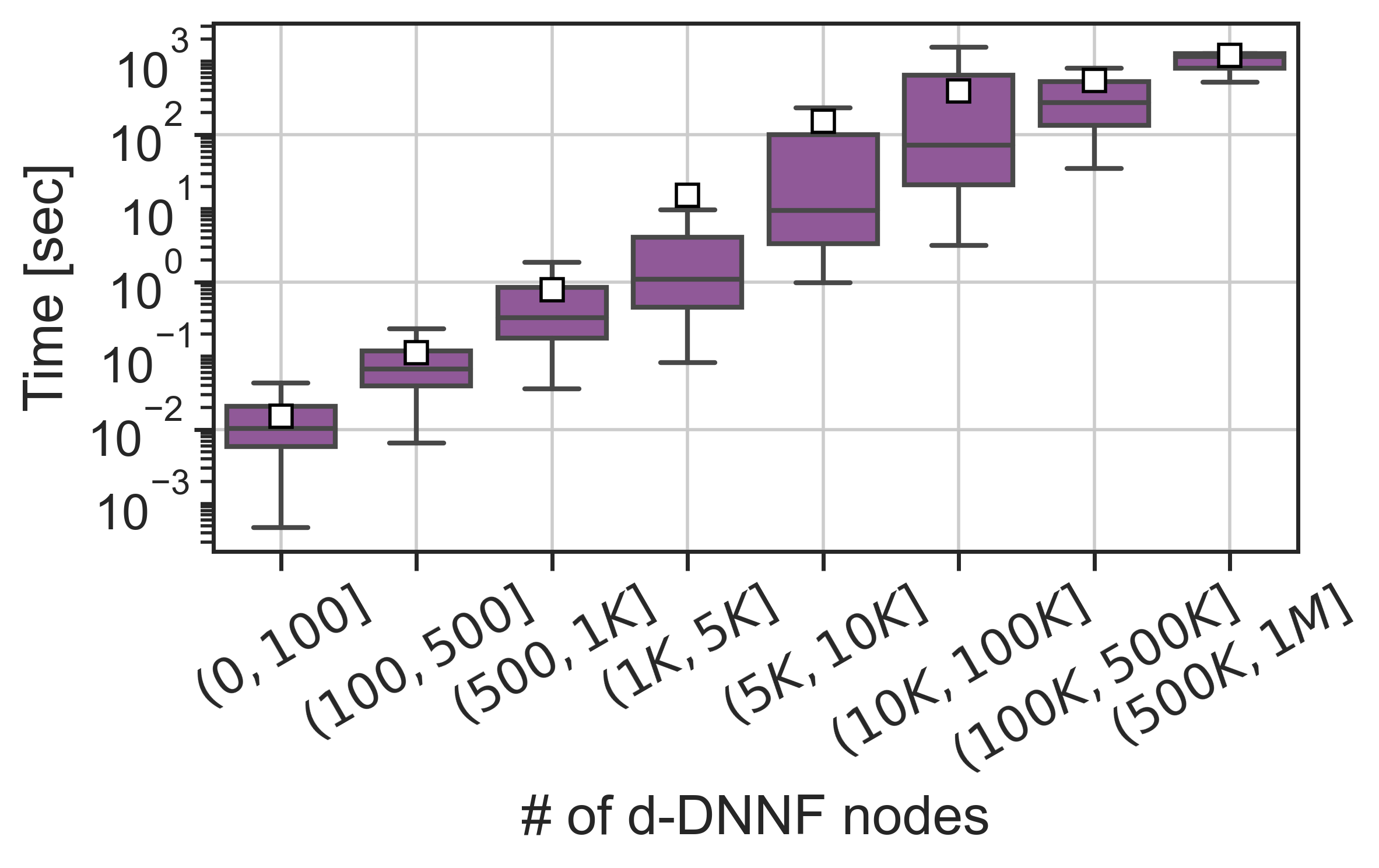}
    \caption{\revision{Alg.~\ref{algo:main} time per d-DNNF size}}
    \label{fig:runnig_time_shapley_ddnnf}
    \end{subfigure}
    \caption{\revision{Running times of knowledge compilation and computation of Shapley values from the d-DNNF as function of the number of distinct facts, CNF clauses, and d-DNNF size}}
    \label{fig:running_time_factors}
    \vspace{-1em}
\end{figure}

\paragraph{Execution time.}
For each of the 40 queries, we have measured the execution time of each step
of the computation. First, we have measured (in the column “Execution time”)
the execution time in PostgreSQL, which includes provenance generation for every
output tuple using ProvSQL. Then, for each tuple $\bar{t}$ in the output of the query, we
have measured the KC execution time and the execution time of
Algorithm~\ref{algo:main} to compute the contribution of all input facts with respect to the output tuple $\bar{t}$.
For the latter two algorithms, the execution
times varied significantly for the different output tuples, and thus we report
the execution times for different percentiles (mean, p25, p50, p75 and p99).
Observe that the computation is typically efficient; outliers include q11d for
which the execution time of Algorithm~\ref{algo:main} was over~96 seconds in average.

Figure \ref{fig:running_time_factors} depicts the running time of the KC step and of the computation of Shapley values from the d-DNNF as a function of different features of the provenance, such as the number of facts appearing in the provenance, the number of clauses in its CNF representation, and the number of gates in its d-DNNF representation.

\revision{
\paragraph{Scalability.}
To evaluate the scalability of Algorithm~\ref{algo:main} we have further looked at different scales of the TPC-H database.
Figure~\ref{fig:scalability} depicts the running time of Algorithm~\ref{algo:main} over 8 representative query outputs as a function of the number of facts in the \textsc{lineitem} table. Figure~\ref{fig:scalability_easy} depicts 4 representative query outputs, for which the running time takes a few milliseconds, e.g., for ``Q3 23426'' (result 23426 of query $Q3$) and the full TPC-H dataset, the computation of Shapley values for all relevant input facts complete in 4.3ms. In contrast, Figure~\ref{fig:scalability_hard} depicts 4 query outputs for which the exact computation failed to complete over the full TPC-H database. We observe that the algorithm does succeed in these cases if we execute the queries over subsets of the input database, though its execution time may still be high: e.g., if we take a “slice” of the \textsc{lineitem} table consisting of 480,097 facts, then computation of the contribution of all input facts w.r.t. ``Q9 ALGERIA'' takes 556sec.
}

\begin{figure}[b]
    \centering
    \begin{subfigure}{.48\linewidth}
    \includegraphics[width=\linewidth]{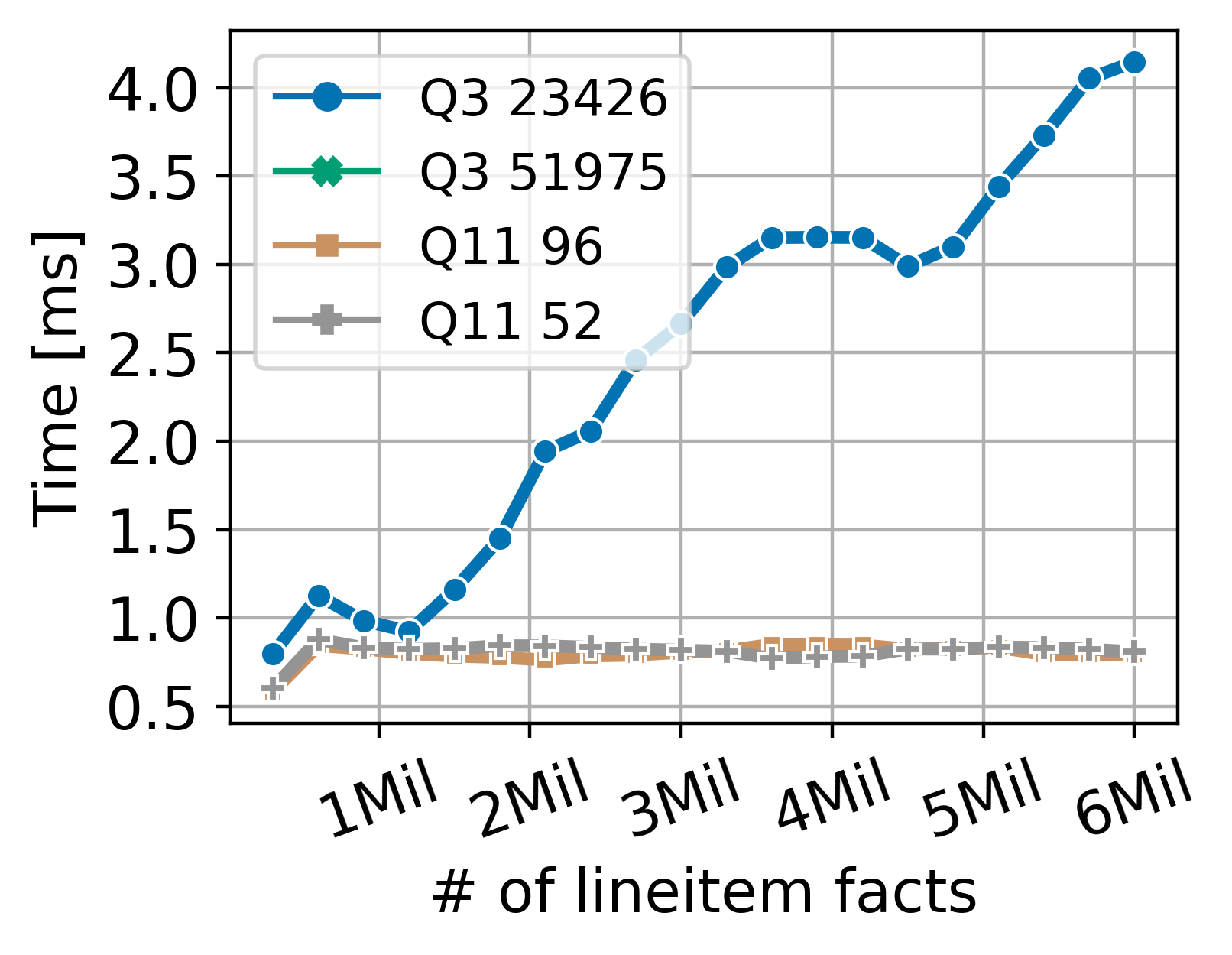}
    \caption{\revision{Representative outputs}}
    \label{fig:scalability_easy}
    \end{subfigure}
    \begin{subfigure}{.48\linewidth}
    \includegraphics[width=\linewidth]{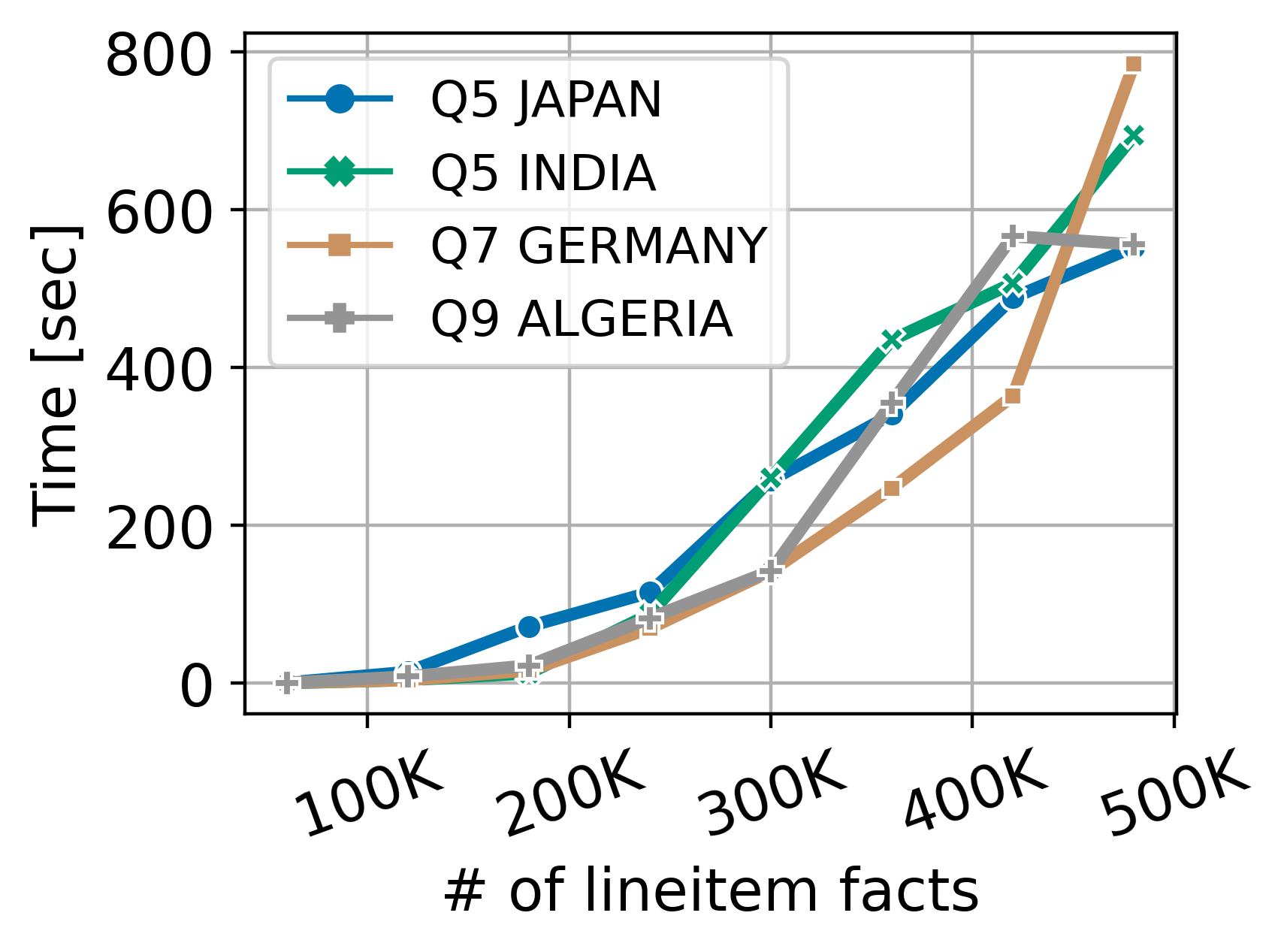}
    \caption{\revision{``Difficult" outputs}}
    \label{fig:scalability_hard}
    \end{subfigure}
    \caption{\revision{Alg.~\ref{algo:main} running time for various TPC-H query outputs as function of table \textsc{lineitem} size}}
    \label{fig:scalability}
\end{figure}

\subsection{Inexact computation}\label{sec:exp_approx}
As observed above, computing exact Shapley values using our solution for a given output tuple is typically fast, but may be costly or even fail in some cases. In this section we evaluate inexact computation alternatives.

\paragraph{Algorithms.}
We compare three algorithms: \cprox{} (Section~\ref{sec:approx}) and two existing baselines: \mc{}, and \kshap{}.
\paragraph*{\textbf{\mc{}}.} 
This is a well-known sampling algorithm \cite{mann1960values} for approximating
Shapley values in general.  To employ the \mc{} algorithm in our setting, we
feed it a provenance expression~$h$ containing~$n$ distinct input facts, and a
budget of~$r\cdot n$ samples, for some $r \in \mathbb{Z}_+$. The Shapley value
of each fact $f$ is approximated by sampling~$r$ permutations ($\pi_1, \ldots,
\pi_r$) of the input facts, and then outputting $\frac{1}{r}\cdot \sum_{i=1}^r
\left( h(S_{\pi_i,<f} \cup \{f\}) - h(S_{\pi_i,<f}) \right)$, where
$S_{\pi_i,<f}$ is the coalition of all facts preceding $f$ in the
permutation~$\pi_i$. 

\paragraph*{\textbf{\kshap{}}} 
Lundberg and Lee~\cite{lundberg2017unified} have defined the notion of SHAP values in the context of ML explainability.
Given a function $h : \R^d \to \R$ (the model whose decisions we want to explain), a probability distribution~$\mathcal{D}$ on~$\R^d$ (the inputs), 
and an input vector $\bar{e} \in \R^d$, SHAP values were defined to measure the contribution of $\bar{e}$'s features to the outcome $h(\bar{e})$.
To overcome the issue that $h$ does not operate over subsets of features, the notion of SHAP-score has been defined as follows
\[\mathsf{SHAP}(h, \bar{e}, x) \defeq \sum_{S\subseteq X\setminus \{x\}} \frac{|S|!(|X|-|S|-1)!}{|X|!} (h_{\bar{e}}(S\cup \{x\}) - h_{\bar{e}}(S)), \]
where $X$ is the set of $d$ features, $x$ is a specific feature whose contribution we wish to assess, and $h_{\bar{e}} : 2^X \to \R$ is defined by $h_{\bar{e}}(S) \defeq \mathbb{E}_{\bar{z} \sim \mathcal{D}}[h(\bar{z}) \mid \bar{e}_S = \bar{z}_S]$, where $\bar{e}_S$ and $\bar{z}_S$ denote the vectors $\bar{e}$ and $\bar{z}$ restricted to the features in $S$.

In \cite{lundberg2017unified} the authors have
proposed a method for approximating SHAP values, called \kshap{}.
\kshap{} assumes feature independence and estimates the probability by
multiplying the marginal distributions $\prod_{i\notin S}\Pr(z_i)$. The
marginal probabilities $\Pr(z_i)$ in turn are estimated from a background data $T$. To
approximate SHAP values, \kshap{} then samples $m$ coalitions $S_1, \ldots,
S_m$ of features,
% (facts in our case), 
and trains a linear model~$g:2^X\to \R$ by minimizing the weighted loss
$\sum_{i=1}^{m} w_i \cdot (g(S_i) - \hat{h}_{\bar{e}}(S_i))^2$, where $w_i$ is
proportional to the size of $S_i$ and $\hat{h}_{\bar{e}}$ is the estimation of $h_{\bar{e}}$
using the feature independence assumption. The coefficient associated with a
feature $x$ in the trained model $g$ is the approximated SHAP value of $x$.

We adapt \kshap{} to our setting and use it to approximate the Shapley values
of facts, as follows. The features are the input facts of the database, we set
$h$ to be the Boolean function representing the (endogenous) provenance; the vector of
interest, $\bar{e}$, has the value~$1$ for all facts, and the background data $T$
contains a single example with value~$0$ in all its entries. Note that for such
$\bar{e}$ and $T$, \kshap{} estimates $h_{\bar{e}}(S)$ as the result of applying $h$ to a
vector with ones in $S$ features and zeros in the remaining entries. 

The input of both \mc{} and \kshap{} includes a budget of $m$ samples; for provenance expressions containing $n$ distinct facts we have experimented with $m \in \{10n,20n,30n,40n,50n\}$.\\

\begin{figure*}[!htb]
    \centering
    \begin{subfigure}{.32\linewidth}
    \centering
    \includegraphics[width=\linewidth]{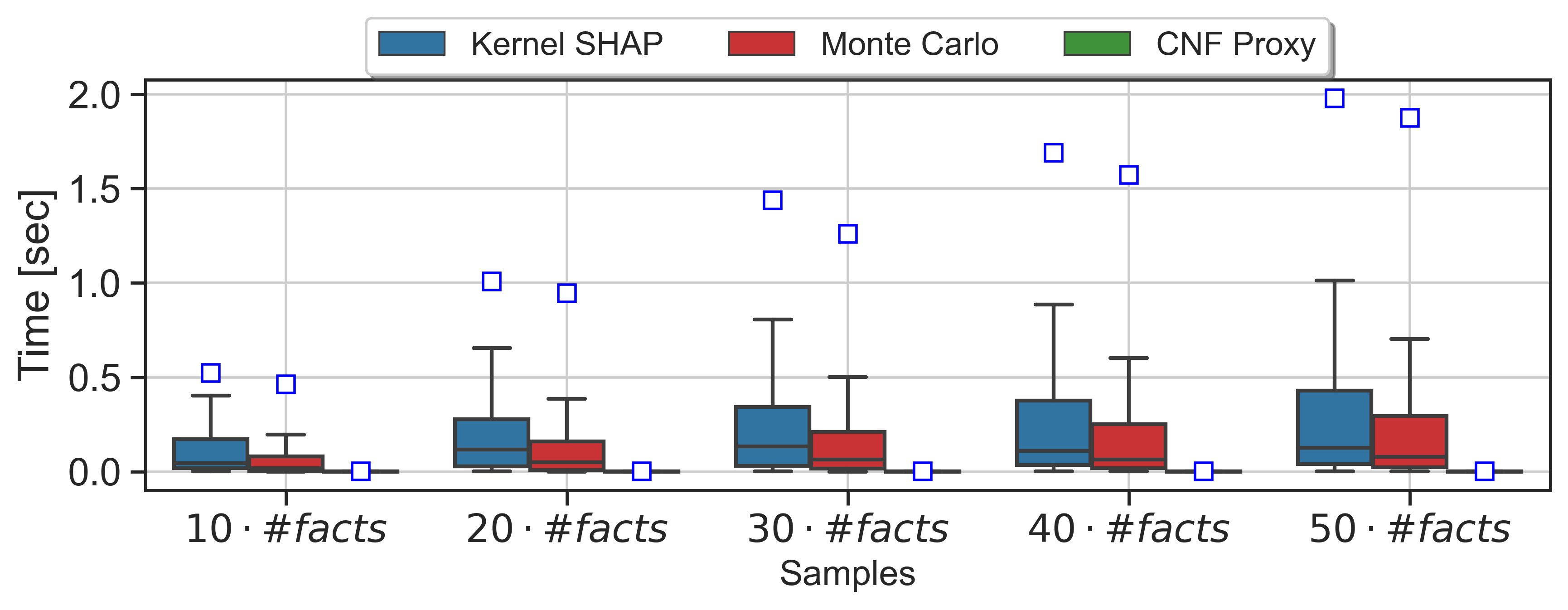}
    \caption{\revision{Execution time}}
    \label{fig:approximated_execution_time}
    \end{subfigure}    
    \begin{subfigure}{.32\linewidth}
    \centering
    \includegraphics[width=\linewidth]{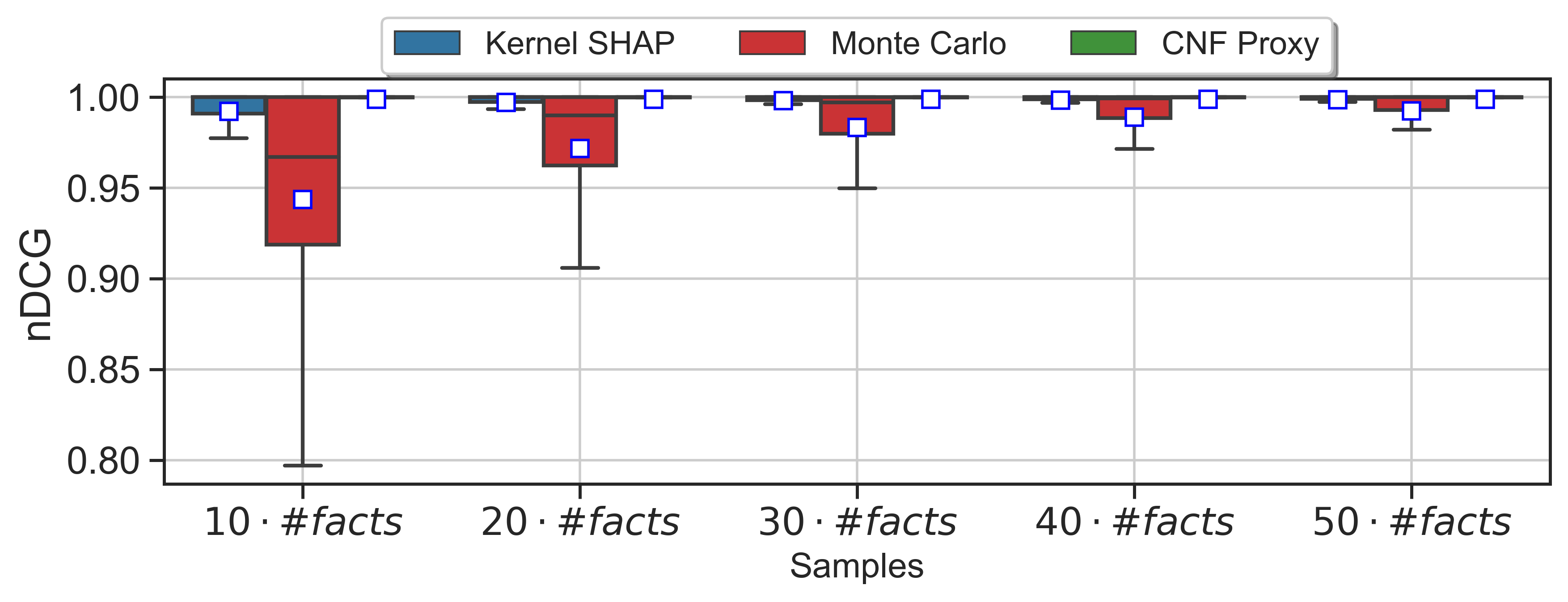}
    \caption{\revision{nDCG}}
    \label{fig:approximated_nDCG}
    \end{subfigure}
    \begin{subfigure}{.32\linewidth}
    \centering
    \includegraphics[width=\linewidth]{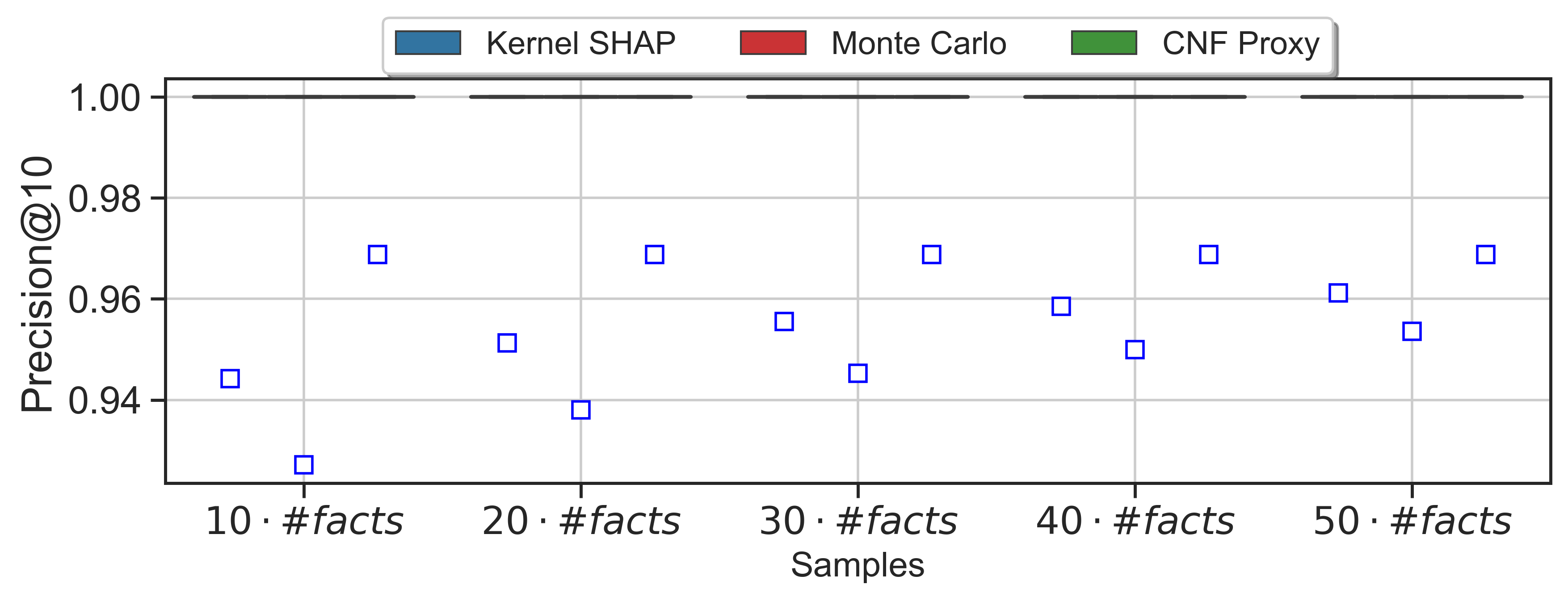}
    \caption{\revision{Precision@10}}
    \label{fig:approximated_precision_at_10}
    \end{subfigure}
    \vspace{-1em}
    \caption{\revision{Comparing various metrics for inexact methods as a function of the sampling budget; \cprox{} does not rely on sampling, thus remains constant for all budgets.}}
    \label{fig:approximated}
\end{figure*}

\begin{table}[t]
    \centering
    \small
    \caption{\revision{Median (resp., mean) performance. \mc{} and \kshap{} use $50\cdot \text{\#facts} $ samples}}
    \label{fig:approximated_table_at_50_samples}
    \vskip-0.5em
    \begin{tabular}{| c | c c c |}
        \hline
         &  \mc{} & \kshap{} & \cprox{}\\
         \hline
         Execution time  & \revision{0.079 (1.875)} & \revision{0.127 (1.978)} & \revision{\textbf{7e-4 (0.002)}}\\
         L1 & \revision{0.439 (0.448)} & \revision{\textbf{0.110 (0.109)}} & \revision{0.317 (0.315)}\\
         L2 & \revision{0.03 (0.034)} & \revision{\textbf{0.001 (0.002)}} & \revision{0.010 (0.014)}\\
         nDCG & \revision{1.0 (0.992)} & \revision{1.0 (0.998)} & \revision{\textbf{1.0 (0.999)}}\\
        %  Precision@3 & 1.0 (0.9882) & 1.0 (0.9951)  & \textbf{1.0 (0.9974)}\\
         Precision@5 & \revision{1.0 (0.955)} & \revision{1.0 (0.961)}  & \revision{\textbf{1.0 (0.989)}}\\
         Precision@10 & \revision{1.0 (0.953)} & \revision{1.0 (0.961)}  & \revision{\textbf{1.0 (0.968)}}\\
         \hline
    \end{tabular}
    \vskip-1em
\end{table}    

\smallskip

\textit{Accuracy metrics}
To evaluate the performance of the above methods we have used various
metrics, specified below. All metrics were computed with respect to the ground
truth values obtained by the knowledge compilation approach\revision{, and thus these experiments are confined to the cases where the exact computation succeeded.}

\begin{compactitem}
\item \paratitle{nDCG} is the normalized discounted cumulative gain score \cite{wang2013theoretical}, used to compare the ordering based on the inexact solution to the ordering based on the ground truth.
\item \paratitle{Precision@k} is the number of facts that appears in the top-$k$ of both the inexact and exact solutions, divided by $k$. This was evaluated for $k \in \{1, 3, 5, 10\}$.
\item \paratitle{L1} and \paratitle{L2} are the mean absolute error and squared error, respectively, of the results of an inexact computation method with respect to the ground truth (i.e. how different the results are from the actual Shapley values).

\end{compactitem}

\begin{figure}[t]
    \begin{subfigure}{.49\linewidth}
    \centering
    \includegraphics[width=\linewidth]{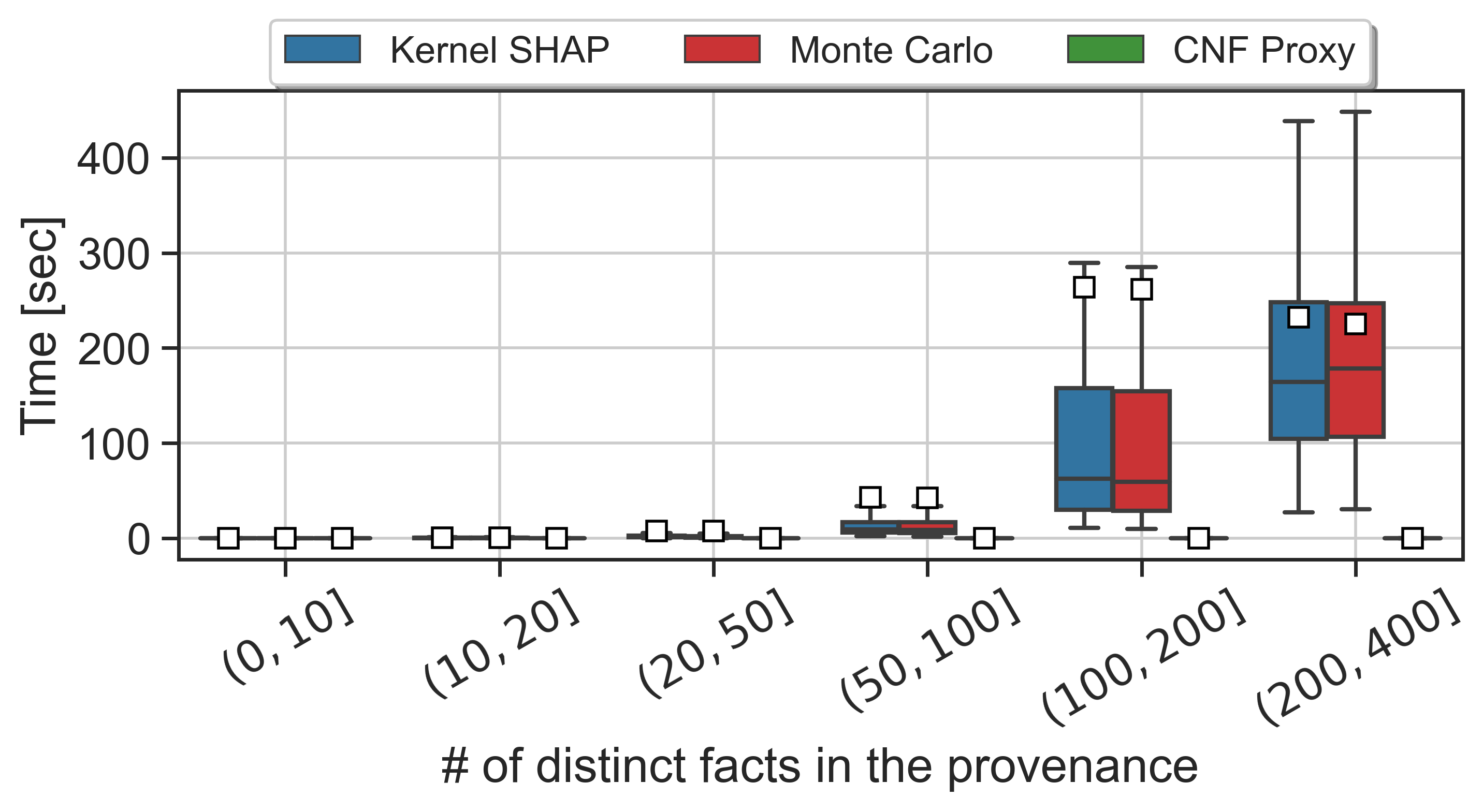}
    \caption{\revision{Running time distribution}}
    \label{fig:approx_methods_running}
    \end{subfigure}
    \begin{subfigure}{.49\linewidth}
    \centering
    \includegraphics[width=\linewidth]{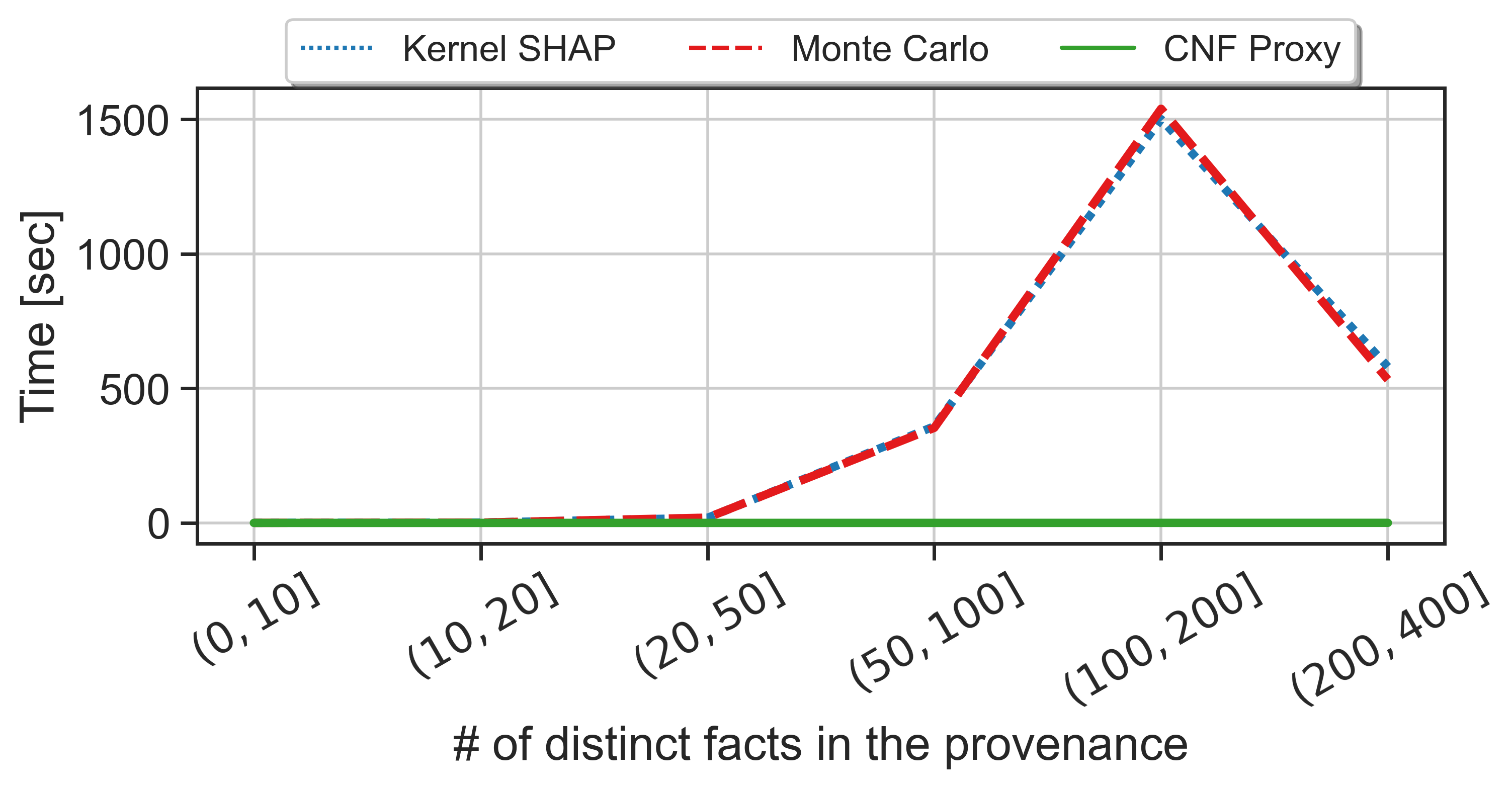}
    \caption{\revision{Worst case running time}}
    \label{fig:approx_methods_running_time_worse_case}
    \end{subfigure}    
    \centering
    \begin{subfigure}{.49\linewidth}
    \centering
    \includegraphics[width=\linewidth]{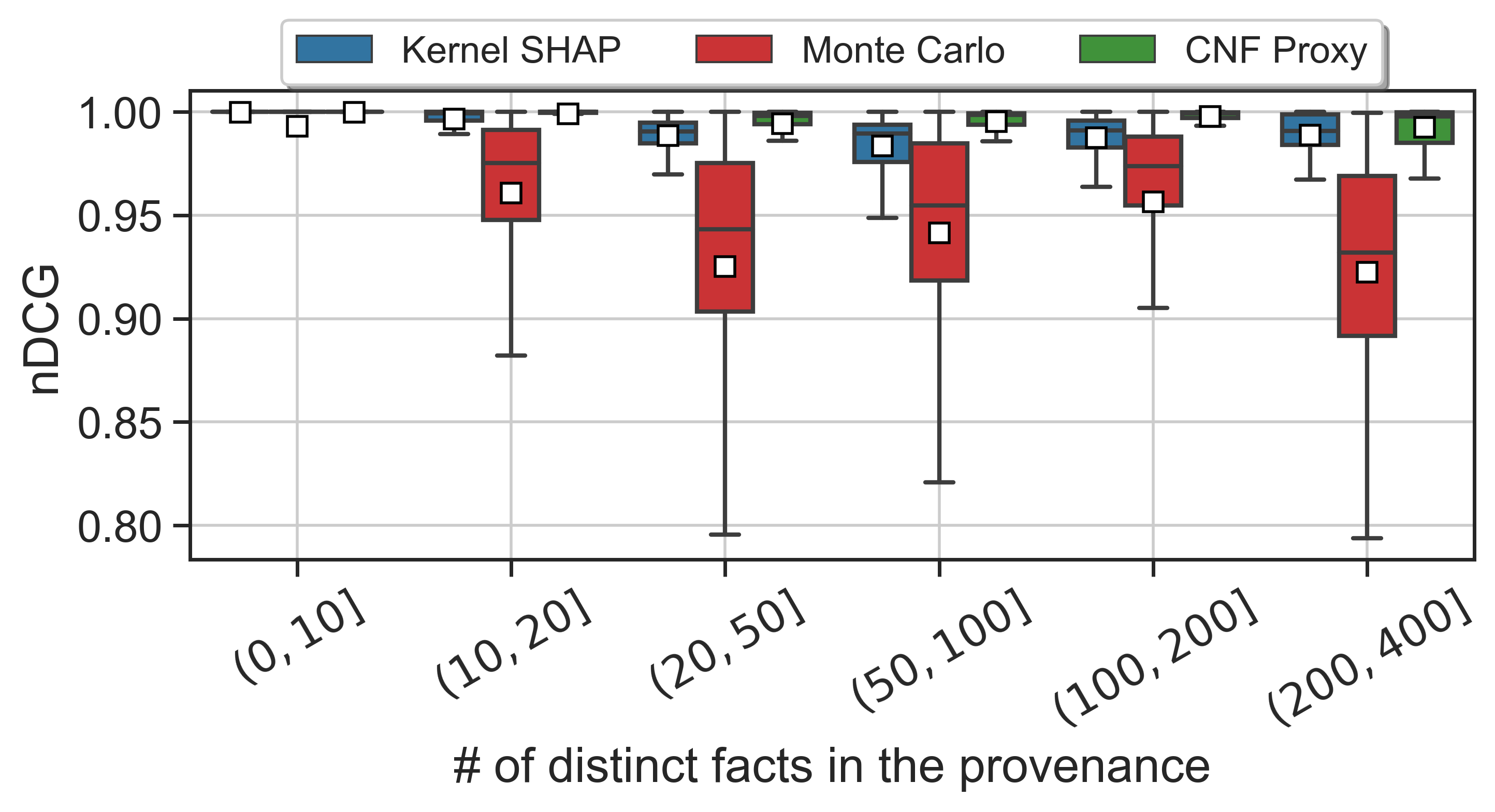}
    \caption{\revision{nDCG Distribution}}
    \label{fig:approx_methods_ndcg}
    \end{subfigure}
    \begin{subfigure}{.49\linewidth}
    \centering
    \includegraphics[width=\linewidth]{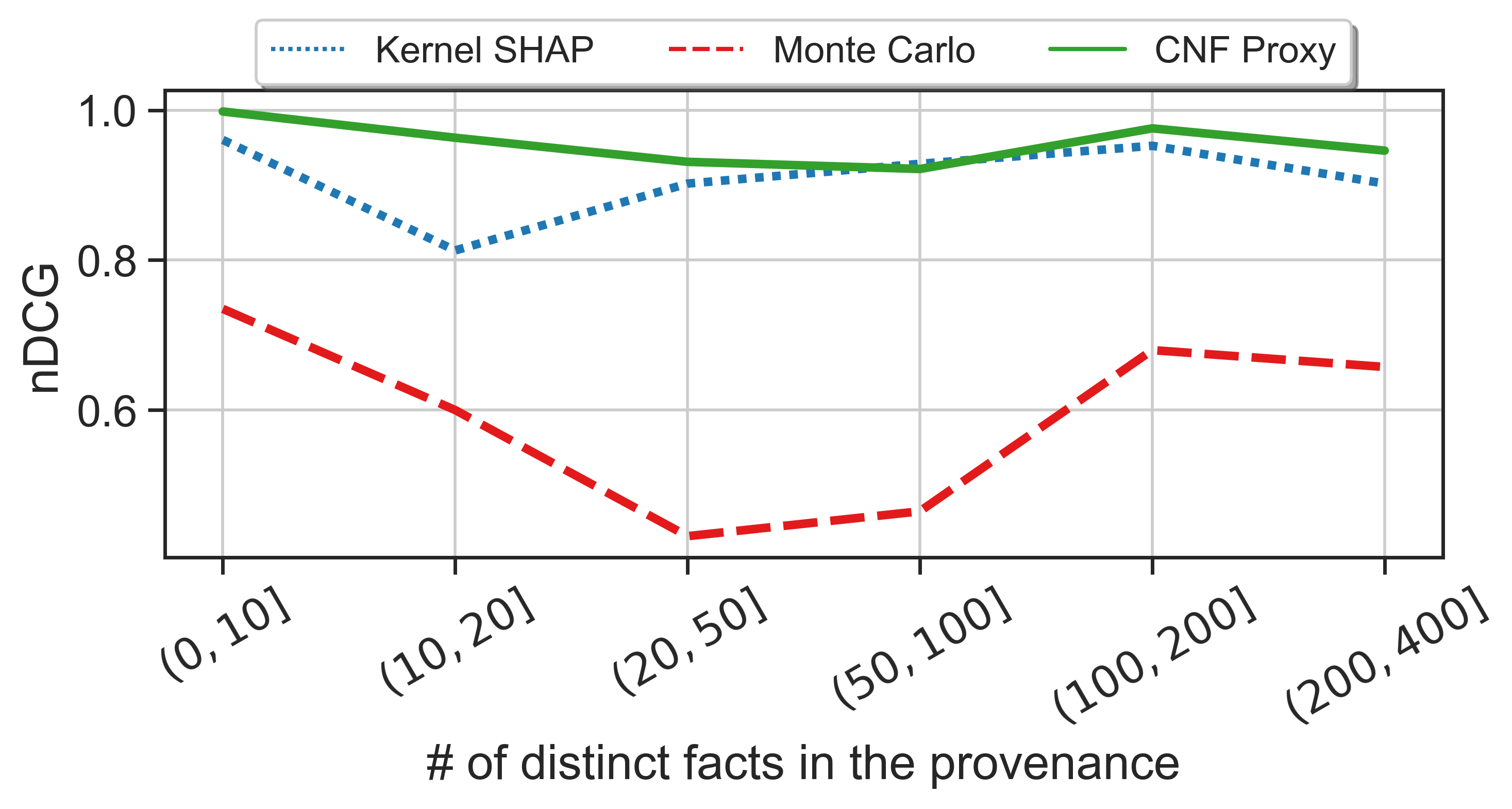}
    \caption{\revision{Worst case nDCG}}
    \label{fig:approx_methods_ndcg_wose_case}
    \end{subfigure}    
    \begin{subfigure}{.49\linewidth}
    \centering
    \includegraphics[width=\linewidth]{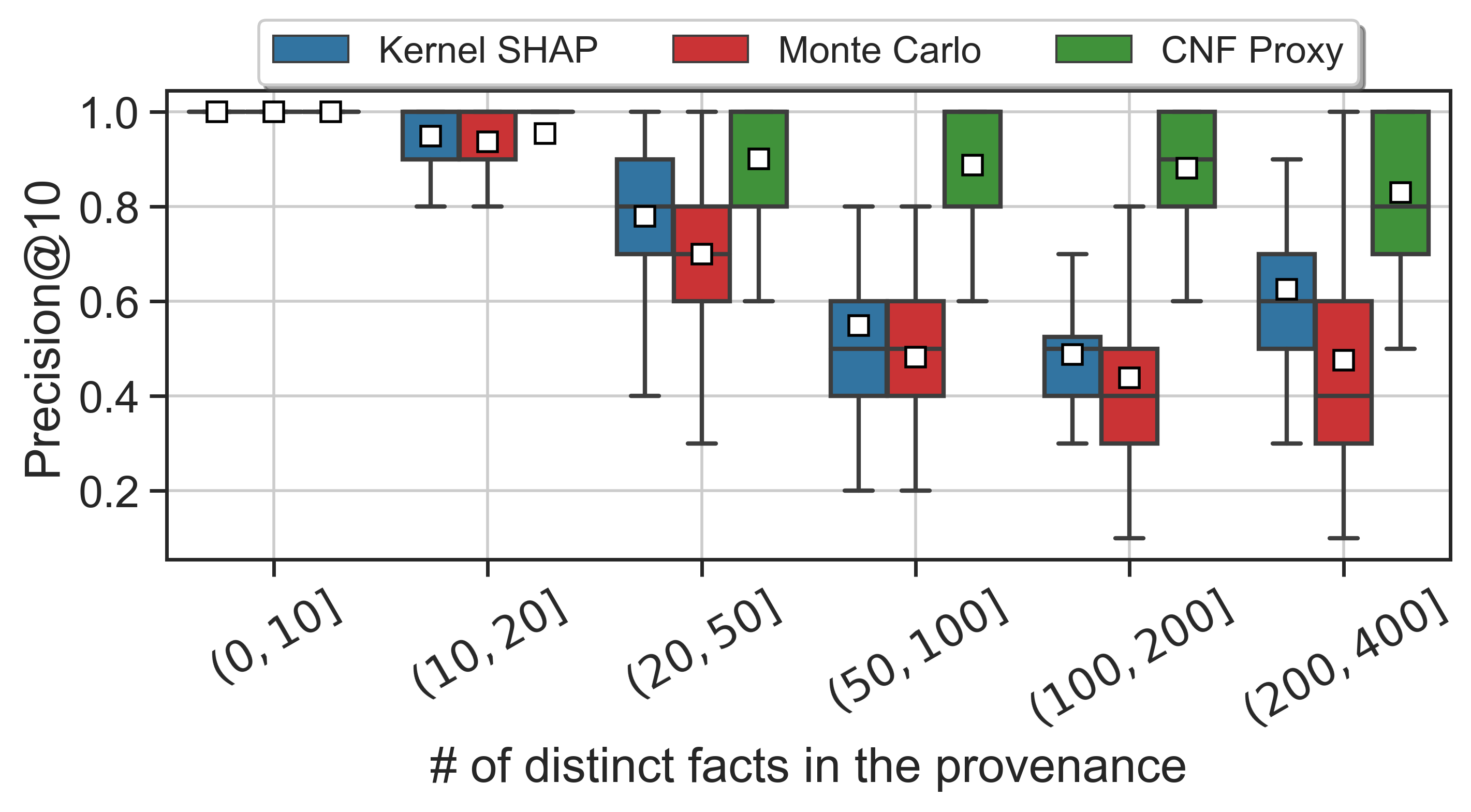}
    \caption{\revision{Precision@10 distribution}}
    \label{fig:approx_methods_precision_at_10}
    \end{subfigure}
    \begin{subfigure}{.49\linewidth}
    \centering
    \includegraphics[width=\linewidth]{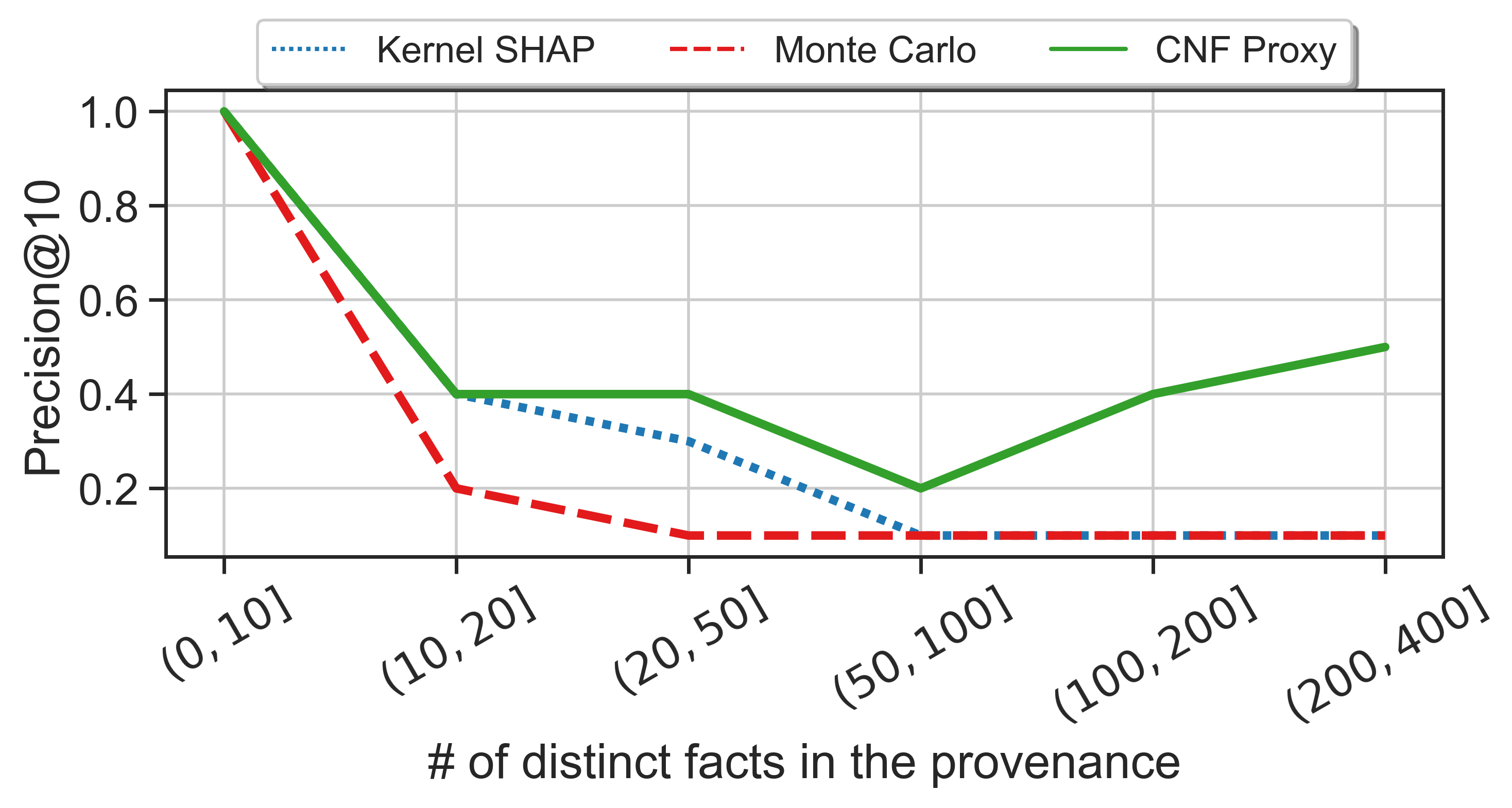}
    \caption{\revision{Worst case Precision@10}}
    \label{fig:approx_methods_precision_at_10_worse_case}
    \end{subfigure}
    \caption{\revision{Comparing inexact methods as a function the number of distinct facts in the provenance ($n$). Sampling methods (\kshap{} and \mc{}) presented with sampling budget of $m=20n$; different budgets led to similar trends.}}
    \label{fig:approx_methods}
    \vspace{-1em}
\end{figure}

\paragraph{Execution time.}
Figure \ref{fig:approximated_execution_time} depicts the execution time of the above methods as a function of the sampling budget. The execution times of \mc{} and \kshap{} are rather similar, while the \cprox{} method (which does not rely on sampling) is substantially faster, with a median of \revision{0.72 milliseconds} and a mean of \revision{2.06 milliseconds} for a single query output.
\revision{Figures~\ref{fig:approx_methods_running} and \ref{fig:approx_methods_running_time_worse_case} depict the distribution and worst-case running times of the three methods (for \mc{} and \kshap{} we used a budget of 20 samples per fact) as function of the number of distinct facts in the provenance expressions. \cprox{} is substantially faster than its competitors: in most cases \cprox{} completes in few milliseconds and 4 seconds in the worst case, whereas \mc{} and \kshap{} median execution time for circuits with 101--200 distinct input facts are 59 and 62 seconds respectively and may take up to 1,539 seconds.}

Recall our analysis of the execution time of the exact solution in Section~\ref{sec:exp_exact}. The computation of Shapley values (KC plus Algorithm~\ref{algo:main}) takes \revision{27,600 seconds} if performed for all output tuples of query q8d in the IMDB dataset; for comparison, the \cprox{} method has (inexactly, see quality analysis below) computed the Shapley value of the proxy functions for all query outputs within \revision{95 seconds, which is 0.3\% of the exact computation time}. Surprisingly, the execution time of the exact computation was comparable to \mc{} and \kshap{}, and even faster for large sampling budgets; For example, \kshap{} with $m=10n$ (where $n$ is the number of distinct facts in the provenance) has completed computation for all output tuples of the query q8d in \revision{16,447 seconds, which is 59\% of the exact computation time}. For $m=50n$, \kshap{} required \revision{58,161 seconds for completion of this computation, which is 210\% of the exact computation time.} This means that using such a sampling budget for \kshap{} is impractical in this setting, and we will use it to obtain upper bounds on \kshap{}'s quality.

\smallskip

\textit{Quality analysis.}
Figures~\ref{fig:approximated_nDCG} and~\ref{fig:approximated_precision_at_10} depict the ranking quality of the above methods as a function of the sampling budget. The rankings were compared for output tuples where the exact computation succeeded, so we have the ground truth. Recall that \cprox{} does not rely on sampling, and so it remains constant throughout the different budgets. 
Naturally, as the sampling budget grows, \mc{} and \kshap{} quality improves. Most notable is \mc{} improvement in terms of nDCG (Figure~\ref{fig:approximated_nDCG}), where its median (resp., mean) nDCG with budget of 10 samples per fact is \revision{0.9669 (resp., 0.9435)}, while with budget of 50 samples per fact it is \revision{1.0 (resp., 0.9923)}. Comparison between the two sampling methods (\mc{} and \kshap{}) reveals that \kshap{} is superior w.r.t all metrics, across the entire budgets range. 
In terms of nDCG,
all methods achieve rather high scores, but \cprox{} performs best. Indeed, the median (resp., mean) nDCG of \cprox{} is \revision{1.0 (resp., 0.9989), while \kshap{} requires 50 samples per fact to get a median (resp., mean) nDCG of 1.0 (resp., 0.9986)}; as previously noted such budget leads to slower execution than that of the exact computation. It is worth mentioning that also when looking at nDCG@k for $k \in \{1, 3, 5, 10\}$, \cprox{} performs better than \mc{} and \kshap{}, even when allowing a budget of~50 samples per fact.
In terms of identifying the top influential facts (i.e., Precision@k), \cprox{} also outperforms the other methods. For example, \cprox{}'s median (resp., mean) Precision@10 are \revision{1.0 (resp., 0.9688)}, while \kshap{} reaches \revision{1.0 (resp., 0.9611)} with~50 samples per fact. Similarly, \cprox{} outperforms \mc{} and \kshap{} in terms of Precision@k for $k\in \{1, 3, 5\}$. 

\revision{
Finally, Table~\ref{fig:approximated_table_at_50_samples} zooms in on the results when fixing the budgets of \mc{} and \kshap{} to 50 samples per fact, which is the highest budget tested (already for this budget, computation of \mc{} and \kshap{} is slower than for the exact algorithm). Note that \cprox{} is much faster then the other methods, while still being superior in terms of ranking (nDCG and Precision@k). As explained above, ranking of facts is indeed the use case we recommend for \cprox{}. Unsurprisingly, \kshap{} achieves better distance from the exact Shapley values (L1 and L2), at the cost of being slower by several orders of magnitude.}

\paragraph{Dependency on the provenance size.}
Figure~\ref{fig:approx_methods} depicts the performance of \mc{}, \kshap{}, and \cprox{} as a function of the number of facts in the provenance expression. \revision{The results are aggregated over all output tuples of all queries.}
Figure~\ref{fig:approx_methods_ndcg} presents the quality of facts ranking (nDCG) as a function of the number of distinct provenance facts. \cprox{} performs the best, and its quality remains steady regardless of the number of facts in the provenance expression. For example, with 1--10 facts the \cprox{} median (resp., mean) nDCG is \revision{1.0 (resp., 0.9999), and with 201--400 facts it is 0.9977 (resp., 0.9924)}. \kshap{} has a minor deterioration, where it drops from median (resp., mean) nDCG of \revision{1.0 (resp., 0.9998)} with 1--10 facts to \revision{0.9906 (resp., 0.9888) with 201--400 facts.} 
\revision{
Figure~\ref{fig:approx_methods_ndcg_wose_case} zooms in on the aggregated results over the worst case expressions of all output tuples for all queries, and shows that even the worst case \cprox{} is superior to the alternatives, and that the error in terms of nDCG is small (0.92 in the worst case evaluated).
}
Figure~\ref{fig:approx_methods_precision_at_10} depicts the dependency of Precision@10 on the number of provenance facts. Both \mc{} and \kshap{} suffer from a massive drop of median (resp., mean) Precision@10, down to \revision{0.4 (resp., 0.476) and 0.6 (resp., 0.6253) respectively with 201--400 facts, while \cprox{} remains at 0.8 (resp., 0.8293) median and mean Precision@10}. A similar trend is observed for Precision@5, whereas for Precision@3 and Precision@1 the drop is less significant.
\revision{
Here again, Figure~\ref{fig:approx_methods_precision_at_10_worse_case} zooms in on the worst case and again shows the superiority of \cprox{}.
}
\subsection{\revision{Hybrid computation}}
\label{sec:exp_hybrid}

Recall that in Section~\ref{sec:exp_exact} we have measured the success rate of
the exact computation: for IMDB the success rate was
99.96\% and for TPC-H the success rate was 84.43\%. We saw in Section~\ref{sec:exp_approx} that the
inexact method \cprox{} is very efficient and that the ranking of tuples based on \cprox{}
is typically close to the ranking obtained based on the real Shapley values.

In this section we consider a hybrid approach that works
as follows. First, we start by running the exact computation, that is, the
knowledge compilation step and Algorithm~\ref{algo:main}.  If the exact
computation completes successfully within less than $t$ seconds (where $t$ is configurable) we return its
result. Otherwise, we terminate the exact computation and execute the inexact
method \cprox{}, returning only a ranking of input tuples rather than their Shapley values.  Figure~\ref{fig:success_rate} depicts the success rate of the
exact computation given different timeouts. Note that given a timeout of 2.5
seconds, the exact computation succeed for 98.67\% of the IMDB output tuples, and
83.83\% for TPCH-H. Increasing the timeout has a rather minor impact on
the success rate: having a 15 seconds timeout increases the success rate for IMDB to
99.52\%, while the success rate for TPC-H remains unchanged.
Recall that having a timeout of one hour results in success rates of 99.96\%
and 84.43\% for IMDB and TPC-H respectively.
Figure~\ref{fig:hybrid_total_time} depicts the mean execution time of the hybrid
approach as a function of the chosen timeout $t$. Observe that given a timeout of 2.5
seconds the mean hybrid execution time is 0.31 seconds and 0.67 seconds for
IMDB and TPC-H respectively. The mean execution time of the hybrid approach grows very moderately w.r.t. the timeout for IMDB (since most cases do not reach the timeout); it grows faster for TPC-H, where the difficult cases for which timeout is reached have a more significant effect on the overall mean execution time.  

\begin{figure}[t]
    \centering
    \begin{subfigure}{.48\linewidth}
    \centering
    \includegraphics[width=\linewidth]{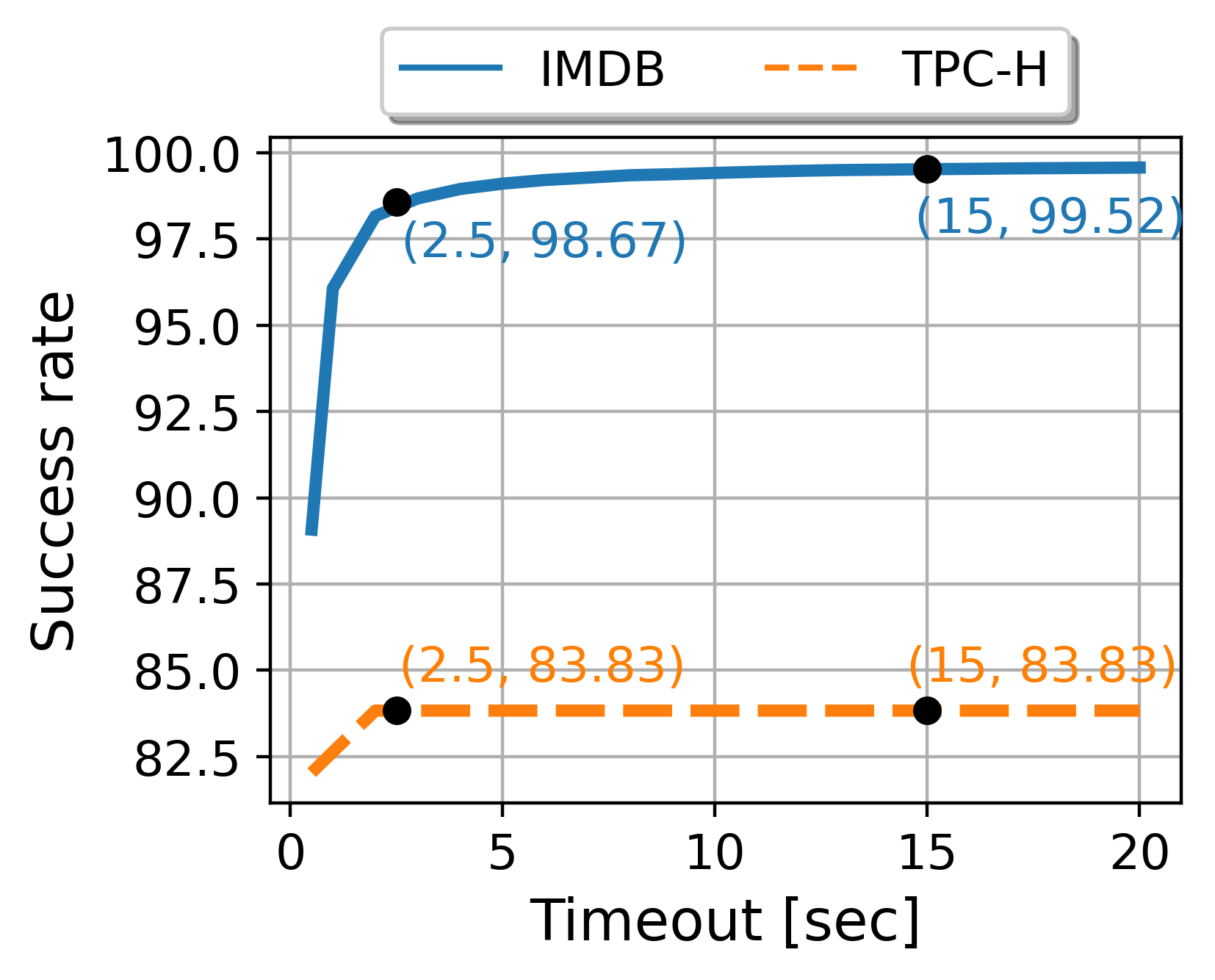}
    \caption{\revision{Exact computation success rate as function of the computation timeout\\}}
    \label{fig:success_rate}
    \end{subfigure}
    \hfill
    \begin{subfigure}{.48\linewidth}
    \centering
    \includegraphics[width=\linewidth]{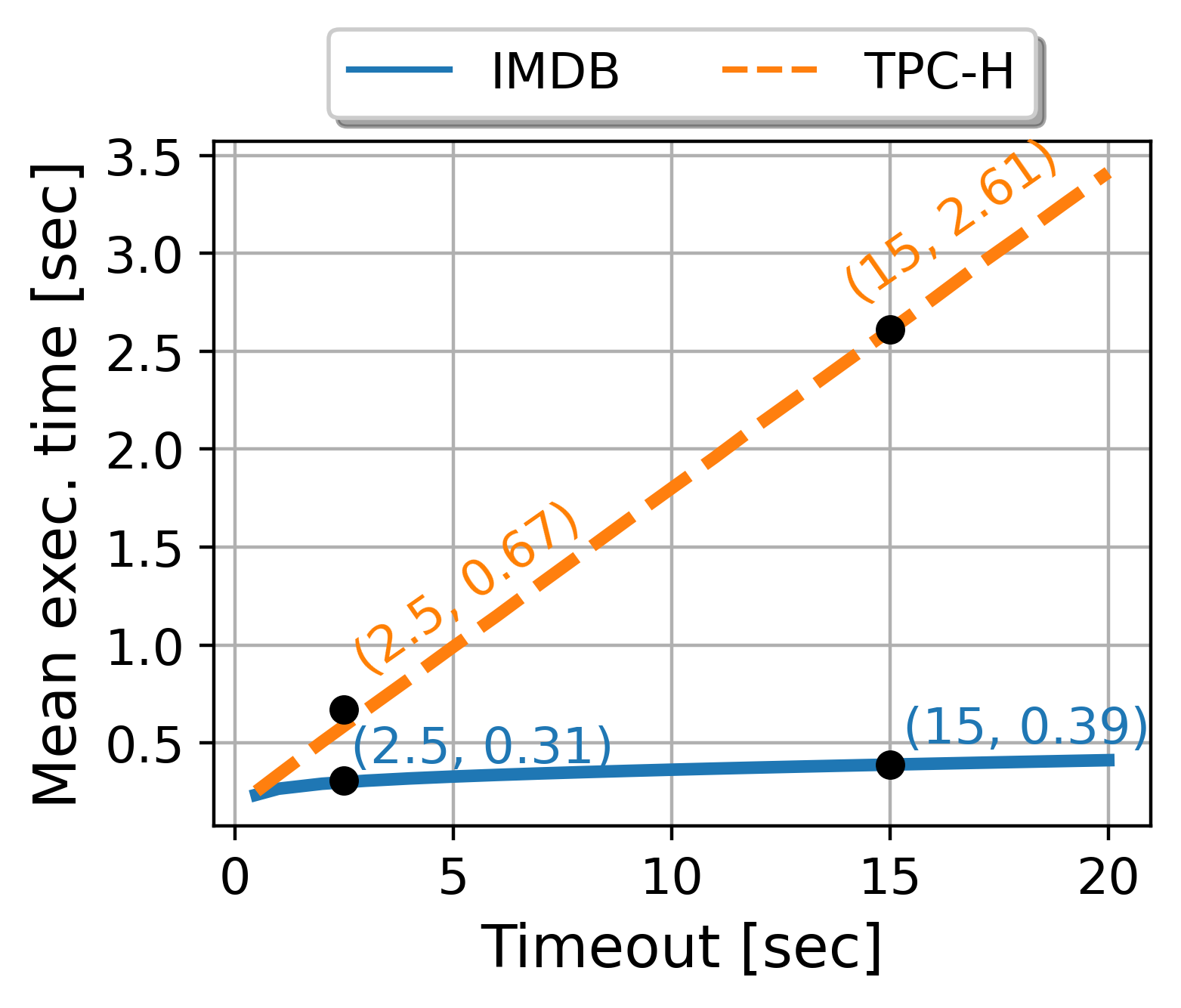}
    \caption{\revision{Mean execution time of the hybrid approach as a function of the computation timeout}}
    \label{fig:hybrid_total_time}
    \end{subfigure}   
    \vspace{-0.5em}
    \caption{\revision{Hybrid approach performance}}
    \label{fig:hybrid_approach}
    \vspace{-1.5em}
\end{figure}

\smallskip

\textit{Main conclusions.\,} Our experimental results indicate that for most output tuples (98.67\% for IMDB and 83.83\% for TPC-H) exact computation of Shapley values terminates within 2.5 seconds. 
When it does not, we propose an alternative of ranking facts according to \cprox{} values, which typically only takes several milliseconds (up to 4 seconds for an outlier case). Experimental evidence  shows that the ranking obtained via \cprox{} is both much faster to compute and more accurate (in terms of nDCG and Precision@k, measured in cases where exact computation does succeed, and so we have the ground truth) than the alternative of only using sampling to approximate the actual Shapley values.    

\vspace{-1mm}
\section{Conclusion}
\label{sec:conc}

We have proposed in this paper a first practical framework for computing the
contributions of database facts in query answering, quantified through Shapley
values. The framework includes an exact algorithm that computes the
contribution of input facts, and a faster algorithm that is practically
effective in ranking contributions of input facts, while producing inexact
Shapley values. Our practical implementation is currently designed for SPJU
queries
(that is, to the class of queries supported by ProvSQL).  In addition to these
practical contributions, we have also established a theoretical connection
between the problem of computing Shapley values and that of probabilistic query
evaluation, by showing that, for every query, the former can be reduced in
polynomial time to the latter.  We leave it open to determine whether there is
also a reduction in the other direction (Open Problem~\ref{open:pqe-to-shap}).
Other interesting directions would be to study further constructs such as
aggregates and negation, or to extend the framework to bag
semantics. Concerning bag semantics we observe that, by differentiating each copy of a same
tuple in a bag database (for instance, adding an identifier attribute), our
framework can be used as-is.  Nevertheless, it would be
interesting to see how one could adapt the definitions in order to consider
fact multiplicities in a more elaborate way.

\begin{acks}
This research has been partially funded by the European Research Council (ERC) under the European Union’s Horizon
2020 research and innovation programme (Grant agreement No. 804302).
The work of Benny Kimelfeld was supported by the Israel Science Foundation (ISF), Grant 768/19, and the German Research Foundation (DFG) Project 412400621 (DIP program).
\end{acks}

\bibliographystyle{ACM-Reference-Format}
\bibliography{sample}

\end{document}